\documentclass[11pt]{article}

\usepackage[english]{babel}
\usepackage[utf8]{inputenc}
\usepackage[T1]{fontenc}

\usepackage[a4paper,top=3cm,bottom=2cm,left=3cm,right=3cm,marginparwidth=1.75cm]{geometry}

\usepackage{amsmath}
\usepackage{amsthm,amssymb}
\usepackage{graphicx}
\usepackage{tikz}
\usepackage[colorlinks=true, allcolors=blue]{hyperref}
\usepackage{xcolor}
\usepackage{cite}
\usepackage{cleveref}
\usepackage{upgreek} 
\usepackage{multirow}
\usepackage{booktabs}
\usepackage{tablefootnote}
\usepackage{tablefootnote,threeparttable}

\newcommand{\norm}[1]{\left\lVert#1\right\rVert}
\newcommand{\norms}[1]{\Big\lVert#1\Big\rVert}

\newcommand{\etal}{\textit{et al}.}
\newcommand{\mc}{\textsf{Monochromatic-Clique}}
\newcommand{\smc}{\textsf{Strong Monochromatic-Clique}}
\newtheorem{theo}{Theorem}[section]
\newtheorem{theorem}[theo]{Theorem}
\newtheorem{lemma}[theo]{Lemma}
\newtheorem{remark}[theo]{Remark}
\newtheorem{definition}[theo]{Definition}
\newtheorem{proposition}[theo]{Proposition}
\newtheorem{corollary}[theo]{Corollary}

\tikzstyle{edge} = [fill,opacity=.5,fill opacity=.5,line cap=round, line join=round, line width=30pt]
\title{Almost Optimal Inapproximability of Multidimensional Packing Problems}
\author{Sai Sandeep\thanks{{\tt spallerl@andrew.cmu.edu}. Research supported in part by NSF grants CCF-1563742 and CCF-1908125.}}
\date{Computer Science Department \\ Carnegie Mellon University \\ Pittsburgh, PA 15213}
\begin{document}
\maketitle
\thispagestyle{empty}
\begin{abstract}
	Multidimensional packing problems generalize the classical packing problems such as Bin Packing, Multiprocessor Scheduling by allowing the jobs to be $d$-dimensional vectors. While the approximability of the scalar problems is well understood, there has been a significant gap between the approximation algorithms and the hardness results for the multidimensional variants. In this paper, we close this gap by giving almost tight hardness results for these problems.
	\begin{enumerate}
		\item We show that Vector Bin Packing has no $\Omega( \log d)$ factor asymptotic approximation algorithm when $d$ is a large constant, assuming $\textsf{P}\neq \textsf{NP}$. This matches the $\ln d + O(1)$ factor approximation algorithms (Chekuri, Khanna SICOMP 2004, Bansal, Caprara, Sviridenko SICOMP 2009, Bansal, Eli\'{a}\v{s}, Khan SODA 2016) upto constants. 
		\item We show that Vector Scheduling has no polynomial time algorithm with an approximation ratio of $\Omega\left( (\log d)^{1-\epsilon}\right)$ when $d$ is part of the input, assuming $\textsf{NP}\nsubseteq \textsf{ZPTIME}\left( n^{(\log n)^{O(1)}}\right)$. This almost matches the $O\left( \frac{\log d}{\log \log d}\right)$ factor algorithms(Harris, Srinivasan JACM 2019, Im, Kell, Kulkarni, Panigrahi SICOMP 2019). We also show that the problem is NP-hard to approximate within $(\log \log d)^{\omega(1)}$. 
		\item We show that Vector Bin Covering is NP-hard to approximate within $\Omega\left( \frac{\log d}{\log \log d}\right)$ when $d$ is part of the input, almost matching the $O(\log d)$ factor algorithm (Alon \etal, Algorithmica 1998). 
	\end{enumerate} 
	Previously, no hardness results that grow with $d$ were known for Vector Scheduling and Vector Bin Covering when $d$ is part of the input and for Vector Bin Packing when $d$ is a fixed constant. 

\end{abstract}

\newpage

\section{Introduction}
\label{sec:intro}
Bin Packing and Multiprocessor Scheduling (also known as Makespan Minimization) are some of the most fundamental problems in Combinatorial Optimization. 
They have been studied intensely from the early days of approximation algorithms and have had a great impact on the field. 
These two are packing problems where we have $n$ jobs with certain sizes, and the objective is to pack them into bins efficiently. In Bin Packing, each bin has unit size and the objective is to minimize the number of bins, while in Multiprocessor Scheduling, we are given a fixed number of bins and the objective is to minimize the maximum load in a bin. 
A relatively less studied but natural problem is Bin Covering, where the objective is to assign the jobs to the maximum number of bins such that in each bin, the load is at least $1$.
These problems are well understood in terms of approximation algorithms: all the three problems are NP-hard, and all of them have a Polynomial Time Approximation Scheme (PTAS)~\cite{VegaL81, HochbaumS87, CsirikJK01}. 

In this paper, we study the approximability of the multidimensional generalizations of these three problems. The three corresponding problems are Vector Bin Packing, Vector Scheduling, and Vector Bin Covering. Apart from their theoretical importance, these problems are widely applicable in practice~\cite{spieksma1994branch, ShachnaiT12,panigrahy2011heuristics} where the jobs often have multiple dimensions such as CPU, Hard disk, memory, etc.

In the Vector Bin Packing problem, the input is a set of $n$ vectors in $[0,1]^d$ and the goal is to partition the vectors into the minimum number of parts such that in each part, the sum of vectors is at most $1$ in every coordinate. 
The problem behaves differently from Bin Packing even when $d=2$: Woeginger~\cite{Woeginger97} proved\footnote{Very recently, Ray~\cite{ray2021aptas} found an oversight in Woeginger's original proof and gave a revised APX hardness proof for the problem.} that there is no asymptotic\footnote{The asymptotic approximation ratio (formally defined in~\Cref{sec:prelims}) of an algorithm is the ratio of its cost and the optimal cost when the optimal cost is large enough. All the approximation factors mentioned in this paper for Vector Bin Packing are asymptotic.} PTAS for $2$-dimensional Vector Bin Packing, assuming $\textsf{P} \neq \textsf{NP}$. 
On the algorithmic front, the PTAS for Bin Packing~\cite{VegaL81} easily implies a $d+\epsilon$ approximation for Vector Bin Packing. 
When $d$ is part of the input, this is almost tight: there is a lower bound of $d^{1-\epsilon}$ shown by~\cite{CK04}\footnote{\cite{CK04} actually give $d^{\frac 12 -\epsilon}$ hardness, but it has been shown later (see e.g., ~\cite{ChristensenKPT17}) that a slight modification of their reduction gives $d^{1-\epsilon}$ hardness.}.
When $d$ is a fixed constant\footnote{The algorithms are now allowed to run in time $n^{f(d)}$, for some function $f$.}, much better algorithms are known~\cite{CK04,BansalCS09,BEK06} that get $\ln d+O(1)$ approximation guarantee. However, the best hardness factor (for arbitrary constant $d$) is still the APX-hardness result of the $2$-dimensional problem due to Woeginger from 1997. 

Closing this gap, either by obtaining a $O(1)$ factor algorithm or showing a hardness factor that is a function of $d$, has remained a challenging open problem.
It is one of the ten open problems in a recent survey on multidimensional scheduling problems~\cite{ChristensenKPT17}. It also appeared in a recent report by Bansal~\cite{Bansal17} on open problems in scheduling. 
In fact, to the best of our knowledge, super constant integrality gap instances for the configuration LP relaxation of the problem were also not known. For the integer instances i.e. when the vectors are from $\{0,1\}^d$ (which are the hard instances for Vector Scheduling and Vector Bin Covering), there is an asymptotic PTAS since there are $O_d(1)$ item types. 
%More generally, we also have an asymptotic PTAS for the problem when the instances have constant Dilworth number~\cite{VegaL81}. 

In the Vector Scheduling problem, given a set of $n$ vector jobs in $[0,1]^d$, and $m$ identical machines, the objective is to assign the jobs to machines to minimize the maximum $\ell_{\infty}$ norm of the load on the machines.
Chekuri and Khanna~\cite{CK04} introduced the problem as a natural generalization of Multiprocessor Scheduling and obtained a PTAS for the problem when $d$ is a fixed constant. 
When $d$ is part of the input, they obtained a $O(\log^2 d)$ factor approximation algorithm. 
They also showed that it is NP-hard to obtain a $C$ factor approximation algorithm for the problem, for any constant $C$.
Meyerson, Roytman, and Tagiku~\cite{MRT13} gave an improved $O(\log d)$ factor algorithm while the current best factor is $O\left( \frac{\log d}{\log \log d}\right)$ due to Harris and Srinivasan~\cite{HarrisS19} and Im, Kell, Kulkarni, and Panigrahi~\cite{IKKP19}. The algorithm of Harris and Srinivasan~\cite{HarrisS19} works for the more general setting of unrelated machines where each job can have a different vector load for each machine. However, no super constant hardness is known even in this unrelated machines setting. 

In the Vector Bin Covering problem, the input is a set of $n$ vectors in $[0,1]^d$. The objective is to partition these into the maximum number of parts such that in each part, the sum of vectors is at least $1$ in every coordinate. 
It is also referred to as ``dual Vector Packing'' in the literature. 
This problem is introduced by Alon \etal ~\cite{AlonACESVW98} who gave a $O(\log d)$ factor approximation algorithm. They also gave a $d$ factor algorithm using a method from the area of compact vector approximation that outperforms the above algorithm for small values of $d$. On the hardness front, Ray~\cite{ray2021aptas} showed that the $2$-dimensional Vector Bin Covering problem is hard to approximate within a factor of $\frac{998}{997}$.
%Woeginger's hardness result~\cite{Woeginger97} for Vector Bin Packing can be easily modified to give APX-hardness for $2$-dimensional Vector Bin Covering as well.  

\subsection{Our Results}
We prove almost optimal hardness results for the three multidimensional problems discussed above. 

\subsubsection{Vector Bin Packing}

For the Vector Bin Packing problem, we prove a $\Omega(\log d)$ asymptotic hardness of approximation when $d$ is a large constant, matching the $\ln d+O(1)$ approximation algorithms~\cite{CK04,BansalCS09,BEK06}, up to constants. 
\begin{theorem}
\label{thm:main-vbp}
There exists an integer $d_0$ and a constant $c>0$ such that for all constants $d \geq d_0$, $d$-dimensional Vector Bin Packing has no asymptotic $c \log d$ factor polynomial time approximation algorithm unless $\emph{\textsf{P}}=\emph{\textsf{NP}}$. 
\end{theorem}

%Before going through our proof approach, we first take a look at challenges in obtaining hardness results for Vector Bin Packing. First, we don't have any candidate hard instances. In fact, to the best of our knowledge, we do not know any super constant integrality gap instances for the configuration LP. The integer instances i.e. when the vectors are from $\{0,1\}^d$ (which are the hard instances for Vector Scheduling and Vector Bin Covering) are easy here since we are operating in the constant $d$ regime. We also have an asymptotic PTAS for the problem when the instances have constant Dilworth number~\cite{VegaL81}. Thus, the hard instances have to be constructed such that the pairs of vectors are incomparable with each other. Woeginger's proof~\cite{Woeginger97} (and Ray's modified analysis~\cite{ray2021aptas}) achieves this using a reduction from the $3$-partition problem and setting the vectors precisely to ensure this. 

%In contrast, we take a conceptually different approach towards showing the hardness of Vector Bin Packing, which is interestingly inspired by the algorithmic progress on the problem. In particular, w
We obtain our hardness result via a reduction from the set cover problem on certain structured instances. 
In the set cover problem, we are given a set system $\mathcal{S}\subseteq 2^V$ on a universe $V$, and the goal is to pick the minimum number of sets from $\mathcal{S}$ whose union is $V$.
Observe that Vector Bin Packing is a special case of the set cover problem with the vectors being the elements and every maximal set of vectors whose sum is at most $1$ in every coordinate (known as ``configurations'') being the sets. In fact, in the elegant Round \& Approx framework~\cite{BansalCS09, BEK06}, the Vector Bin Packing problem is viewed as a set cover instance, and the algorithms proceed by rounding the standard set cover LP. Towards proving the hardness of Vector Bin Packing, we ask the converse: \textit{Which families of set cover instances can be cast as $d$-dimensional Vector Bin Packing?}

We formalize this question using the notion of \emph{packing dimension} of a set system $\mathcal{S}$ on a universe $V$: it is the smallest integer $d$ such that there is an embedding $f:V \rightarrow [0,1]^d$ such that a set $S \subseteq V$ is in $\mathcal{S}$ if and only if 
\[
\norm{\sum_{v \in S}f(v)}_{\infty}\leq 1
\]
If a set system has packing dimension $d$, then the corresponding set cover problem can be embedded as a $d$-dimensional Vector Bin Packing instance.
However, it is not clear if the hard instances of the set cover problem have a low packing dimension. 
Indeed the instances in the $(1-\epsilon)\ln n$ set cover hardness~\cite{Feige98} have a large packing dimension that grows with $n$, which we cannot afford as we are operating in the constant $d$ regime. 
We get around this by starting our reduction from highly structured yet hard instances of set cover. In particular, we study simple bounded set systems which satisfy the following three properties: 
\begin{enumerate}
    \item The set system is \textit{simple}\footnote{Simple set families are also known as linear set families.} i.e., every pair of sets intersect in at most one element. 
    \item The cardinality of each set is at most $k$, a fixed constant. 
    \item Each element in the family is present in at most $\Delta =k^{O(1)}$ sets. 
\end{enumerate}
Kumar, Arya, and Ramesh~\cite{KAR00} proved that simple set cover i.e., set cover with the restriction that every pair of sets intersect in at most one element, is hard to approximate within $\Omega(\log n)$. We observe that by modifying the parameters slightly in their proof, we can obtain the $\Omega(\log k)$ hardness of simple bounded set cover.  

We prove that simple bounded set systems have packing dimension at most $k^{O(1)}$. Thus, the $\Omega(\log k)$ simple bounded set cover hardness translates to $\Omega(\log d)$ hardness of Vector Bin Packing when $d$ is a constant. Note that the optimal value of the set cover instances can be made arbitrarily large in terms of $k$ by starting with a Label Cover instance with an arbitrarily large number of edges. Thus, our Vector Bin Packing hardness holds for asymptotic approximation as well. 

Our upper bound on the packing dimension is obtained in two steps: First, we write the given simple bounded set system as an intersection of $(k\Delta)^{O(1)}$ structured simple bounded set systems on the same universe, and then we give an embedding using $(k\Delta)^{O(1)}$ dimensions bounding the packing dimension of these structured simple bounded set systems. 
This idea of decomposition into structured instances is inspired from a work of Chandran, Francis, and Sivadasan ~\cite{ChandranFS08} where an upper bound on the Boxicity of a graph is obtained in terms of its maximum degree. We believe that the packing dimension of set systems is worth studying on its own, especially in light of its close connections to the notions of dimension of graphs such as Boxicity. 

\subsubsection{Vector Scheduling}
For the Vector Scheduling problem, we obtain a $\Omega\left( (\log d)^{1-\epsilon}\right)$ hardness under $\textsf{NP} \nsubseteq \textsf{ZPTIME}\left(n^{(\log n)^{O(1)}}\right)$, almost matching the $O\left( \frac{\log d}{\log \log d}\right)$ algorithms~\cite{HarrisS19,IKKP19}.
\begin{theorem}
\label{thm:main-vs}
For every constant $\epsilon>0$, assuming $\textsf{NP} \nsubseteq \textsf{ZPTIME}\left(n^{(\log n)^{O(1)}}\right)$, $d$-dimensional Vector Scheduling has no polynomial time $\Omega\left( (\log d)^{1-\epsilon}\right)$-factor approximation algorithm when $d$ is part of the input.
\end{theorem}

We obtain the hardness result via a reduction from the Monochromatic Clique problem. In the \mc(k,B) problem, given a graph $G=(V,E)$ with $|V|=n$ and parameters $k(n)$ and $B(n)$, the goal is to distinguish between the case when $G$ is $k$-colorable and the case when in any assignment of $k$-colors to vertices of $G$, there is a clique of size $B$ all of whose vertices are assigned the same color. 
When $B=2$, this is the standard graph coloring problem.
Note that the problem gets easier as $B$ increases.
Indeed, when $B >\sqrt{n}$, we can solve the problem in polynomial time by computing the Lov\'{a}sz theta function of the complement graph. 
We are interested in proving the hardness of the problem for $B$ as large a function of $n$ as possible, for some $k$. For example, given a graph that is promised to be $k$ colorable, can we prove the hardness of assigning $k$ colors to the vertices of the graph in polynomial time where each color class has maximum clique at most $B=\log n?$

The Monochromatic Clique problem was defined formally by Im, Kell, Kulkarni, and Panigrahi~\cite{IKKP19} in the context of proving lower bounds for online Vector Scheduling. It was also used implicitly in the $\omega(1)$ NP-hardness of Vector Scheduling by Chekuri and Khanna~\cite{CK04}. 
They proved (implicitly) that Monochromatic Clique is NP-hard when $B$ is an arbitrary constant using a reduction from $n^{1-\epsilon}$ hardness of graph coloring. We observe that the same reduction combined with better hardness of graph chromatic number~\cite{Khot01} proves the hardness of Monochromatic Clique when $B=(\log n)^{\gamma}$, for some constant $\gamma>0$ under the assumption that $\textsf{NP}\nsubseteq \textsf{ZPTIME}\left( n^{(\log n)^{O(1)}}\right)$.

We then amplify this hardness to $B=(\log n)^{C}$, for \textit{every} constant $C >0$.
%However, we cannot amplify the hardness directly using some form of graph product as preserving both the $k$-colorability and also amplifying the maximum clique in a $k$-coloring at the same time seems to be impossible in any graph product. 
Our main idea in this amplification procedure is the notion of a stronger form of Monochromatic Clique where given a graph and parameters $k,B,C$, the goal is to distinguish between the case that $G$ is $k$ colorable vs. in any $k^C$ coloring of $G$, there is a monochromatic clique of size $B$. 
It turns out that the graph coloring hardness of Khot~\cite{Khot01} already proves the hardness of this stronger variant of Monochromatic clique when $B=(\log n)^\gamma$ for any constant $C$. We then use lexicographic product of graphs to amplify this result into the hardness of original Monochromatic Clique problem with $B=(\log n)^C$ for any constant $C$ under the same assumption that $\textsf{NP}\nsubseteq \textsf{ZPTIME}\left( n^{(\log n)^{O(1)}}\right)$. 
This directly gives the required hardness of Vector Scheduling using the reduction in ~\cite{CK04}.

\medskip 
The Vector Scheduling problem is also closely related to the Balanced Hypergraph Coloring problem where the input is a hypergraph $H$ and a parameter $k$, and the objective is to color the vertices of $H$ using $k$ colors to minimize the maximum number of times a color appears in an edge. We use this connection to improve upon the NP-hardness of the problem:
\begin{theorem}
\label{thm:main-vs-np}
For every constant $C>0$, $d$-dimensional Vector Scheduling is NP-hard to approximate within $\Omega\left((\log \log d)^C\right)$ when $d$ is part of the input.
\end{theorem}

Consider the case when each vector job is from $\{ 0,1 \}^d$. In this setting, we can view each coordinate as an edge in a hypergraph, and each vector corresponds to a vertex of the hypergraph. The goal is to find a $m$-coloring of vertices of the hypergraph i.e., an assignment of the vectors to $m$ machines to minimize the maximum number of monochromatic vertices in an edge, which directly corresponds to the maximum load on a machine. 

This problem of coloring a hypergraph to ensure that no color appears too many times in each edge is known as Balanced Hypergraph Coloring. 
Guruswami and Lee~\cite{GL18} obtained strong hardness results for this problem when $k$, the uniformity of the hypergraph is a constant, using the Label Cover Long Code framework combined with analytical techniques such as the invariance principle. 
However, when $k$ is super constant, the invariance principle based methods give weak bounds as the soundness of the Label Cover has to be at least exponentially small in $k$. Recently, using combinatorial tools to analyze the gadgets instead of the standard analytical techniques, improvements have been obtained for various hypergraph coloring problems~\cite{Bhangale18, ABP19} in the super-constant inapproximability regime. We follow the same route and use combinatorial tools to analyze the gadgets in the Label Cover Long Code framework and obtain better hardness of Balanced Hypergraph Coloring in the regime of super-constant uniformity $k$. The key combinatorial lemma used in our analysis was proved recently by Austrin, Bhangale, and Potukuchi~\cite{ABP20} using a generalization of the Borsuk-Ulam theorem.

The NP-hardness of Vector Scheduling follows directly from the hardness of Balanced Hypergraph Coloring using the above-described reduction. 
This NP-hardness result uses near-linear size Label Cover hardness results~\cite{MR10, DinurS14}.
By using the standard Label Cover hardness obtained by combining PCP Theorem and Parallel Repetition in the same reduction, we also prove an intermediate result bridging the above two hardness results for Vector Scheduling.
\begin{theorem}
\label{thm:main-vs-intermediate}
There exists a constant $\gamma >0$ such that assuming $\textsf{NP} \nsubseteq \textsf{DTIME}\left( n^{O(\log \log n)}\right)$, $d$-dimensional Vector Scheduling is hard to approximate within $\Omega\left((\log d)^\gamma\right)$ when $d$ is part of the input.
\end{theorem}

\subsubsection{Vector Bin Covering}

For the Vector Bin Covering problem, we show $\Omega\left( \frac{\log d}{\log \log  d}\right)$ hardness, almost matching the $O(\log d)$ factor algorithm~\cite{AlonACESVW98}.
\begin{theorem}
\label{thm:main-vbc}
$d$-dimensional Vector Bin Covering is NP-hard to approximate within $\Omega\left( \frac{\log d}{\log\log d}\right)$ factor when $d$ is part of the input. 
\end{theorem}

Similar to Vector Scheduling, the hard instances for the Vector Bin Covering are when each vector is in $\{0,1\}^d$. Using the same connection as before, we view this problem as a hypergraph coloring problem where each edge of the hypergraph corresponds to a coordinate, and each vertex corresponds to a vector. Assigning the vectors to the bins such that in each bin, the sum is at least $1$ in every coordinate corresponds to assigning colors to vertices of the hypergraph such that in every edge, all the colors appear. Such a coloring of the hypergraph with $k$ colors where all the $k$ colors appear in every edge is known as a $k$-rainbow coloring of the hypergraph. 

Strong hardness results are known for approximate rainbow coloring~\cite{GL18, ABP20, GS20} when $k$ is a constant. While these results give decent bounds in the super constant regime, they proceed via the hardness of Label Cover whose soundness is an inverse polynomial function of $k$. Because of this, in the NP-hard instances, the number of edges is at least doubly exponential in $k$. 
Instead, by losing a factor of $2$ in the hardness, we give a reduction to the approximate rainbow coloring problem from Label Cover with no gap i.e., just a ``Label Coverized'' $3$-SAT instance. 
In these hard instances, the number of edges is single exponential in $k$, proving that it is NP-hard to distinguish the case that a hypergraph with $m$ edges has a rainbow coloring with $\Omega\left(\frac{\log m}{\log \log m}\right)$ colors vs. it cannot be rainbow colored with $2$ colors\footnote{Note that rainbow coloring with $2$ colors is the same as proper $2$ coloring (Property B) of the hypergraphs.}.
This hardness of approximate rainbow coloring gives the required Vector Bin Covering hardness immediately using the earlier mentioned analogy between rainbow coloring and Vector Bin Covering. 

\medskip  
We summarize our results in~\Cref{table:results}.

\begin{table}
	\centering
\begin{tabular}{@{}crrr@{}}\toprule
Problem & Subcase & Best Algorithm & Best Hardness \\
\midrule
\multirow{3}{*}{VBP\tnote{Both the algorithms and hardness are considered for asymptotic setting.}} & $d=1$ & PTAS~\cite{VegaL81} & NP-Hard~\cite{GareyJ79} \\ 
&Fixed $d$ & $\ln d+O(1)$~\cite{BEK06} & $\Omega(\log d)$ \\ 
&Arbitrary $d$ & $1+\epsilon d+O\left(\ln \frac{1}{\epsilon}\right)$~\cite{CK04} & $d^{1-\epsilon}$~\cite{CK04,ChristensenKPT17} \\
\cmidrule{1-4}
\multirow{7}{*}{VS} & $d=1$ & EPTAS~\cite{Jansen10} & No FPTAS~\cite{FaigleKT89} \\ 
&Fixed $d$ & PTAS~\cite{CK04} & No EPTAS~\cite{BansalOVZ16} \\ \addlinespace[0.01\textwidth]
& \multirow{5}{*}{Arbitrary $d$} & \multirow{5}{*}{$O\left( \frac{\log d}{\log \log d}\right)$~\cite{HarrisS19,IKKP19}} & \multicolumn{1}{l}{$\Omega\left( (\log d)^{1-\epsilon}\right)$}  \\
& & & $\left(\textsf{NP}\nsubseteq \textsf{ZPTIME}\left(n^{(\log n)^{O(1)}}\right)\right)$ \\ \addlinespace[0.01\textwidth]
& & & \multicolumn{1}{l}{$(\log d)^{\Omega(1)}$} \\
& & &  $\left(\textsf{NP}\nsubseteq \textsf{DTIME}\left(n^{O(\log \log n)}\right)\right)$ \\ \addlinespace[0.01\textwidth]
& & & $(\log \log d)^{\omega(1)}$ (NP-Hardness) \\ \addlinespace[0.01\textwidth]
\cmidrule{1-4}
\multirow{2}{*}{VBC} & $d=1$ & FPTAS~\cite{JansenS03} & NP-Hard~\cite{Assman84} \\ 
&Arbitrary $d$ & $O(\log d)$~\cite{AlonACESVW98} & $\Omega\left( \frac{\log d}{\log \log d}\right)$ \\
\bottomrule
\end{tabular}
\caption{Approximation algorithms for the multidimensional packing problems. VBP, VS and VBC stand for Vector Bin Packing, Vector Scheduling and Vector Bin Covering respectively. All the results for VBP and the FPTAS for VBC are asymptotic. The results without citations are our new results.}
\label{table:results}
\end{table}

\subsection{Related Work}

\smallskip \noindent \textbf{Online Algorithms.} Multidimensional packing problems have been extensively studied in the online setting. For the $d$-dimensional Vector Bin Packing, the classical First-Fit algorithm~\cite{GareyGJ76} gives $O(d)$ competitive ratio, and Azar, Cohen, Kamara, and Shepherd~\cite{AzarCKS13} recently gave an almost matching $\Omega\left( d^{1-\epsilon}\right)$ lower bound. For the $d$-dimensional Vector Scheduling, Im, Kell, Kulkarni, and Panigrahi ~\cite{IKKP19} gave a $O\left( \frac{\log d}{\log \log d}\right)$ competitive online algorithm and proved a matching lower bound. For the $d$-dimensional Vector Bin Covering problem, Alon et al.~\cite{AlonACESVW98} gave a $2d$ competitive algorithm and proved a lower bound of $d + \frac 12$. 

\smallskip \noindent \textbf{Geometric variants.} There are various natural geometric variants of Vector Bin Packing that have been studied in the literature. A classical problem of this sort is the $2$-dimensional Geometric Bin Packing, where the input is a set of rectangles that need to be packed into the minimum number of unit squares. After a long line of works, Bansal and Khan~\cite{BansalK14} gave a $1.405$ factor asymptotic approximation algorithm for the problem. On the hardness front, Bansal and Sviridenko~\cite{BansalS04} showed that the problem does not admit an asymptotic PTAS, and this APX hardness result has been generalized to several related problems by Chleb\'{i}k and Chleb\'{i}kov\'{a}~\cite{ChlebikC06}. We refer the reader to the excellent survey~\cite{ChristensenKPT17} regarding the geometric problems.

\subsection{Organization}

We first define the multidimensional problems and the Label Cover problem formally in~\Cref{sec:prelims}. Next, we prove the hardness results for Vector Bin Packing, Vector Scheduling, and Vector Bin Covering in~\Cref{sec:vbp},~\Cref{sec:vs} and~\Cref{sec:vbc} respectively. 
Finally, we conclude by mentioning a few open problems in~\Cref{sec:conclusion}.
\section{Preliminaries}
\label{sec:prelims}

\smallskip \noindent \textbf{Notations.} We use $[n]$ to denote the set $\{1,2,\ldots,n\}$. We use $\textbf{1}^d$ to denote the $d$-dimensional vector $(1,1,\ldots,1)$. For two $d$-dimensional real vectors $\textbf{a}$ and $\textbf{b}$, we say that $\textbf{a} \geq \textbf{b}$ if $\textbf{a}_i \geq \textbf{b}_i$ for all $i \in [d]$. 
For a graph $G$, we let $\omega(G), \alpha(G),\chi(G)$ be the largest clique size,  largest independent size, and the chromatic number of $G$ respectively. A set system or set family $\mathcal{S}$ on a universe $V$ is a collection of subsets of $V$. 

\smallskip \noindent \textbf{Problem Statements.}
We give formal definitions of the three problems that we study. 

\begin{definition}(Vector Bin Packing) In the Vector Bin Packing problem, the input is a set of $n$ rational vectors $v_1, v_2, \ldots, v_n \in [0,1]^d$. The objective is to partition $[n]$ into minimum number of parts $A_1, A_2, \ldots, A_m$ such that 
\[
\Big \lVert\sum_{j \in A_i}v_j \Big \rVert_{\infty} \leq 1 \,\,\forall i \in [m]
\]

\end{definition}

\begin{definition}(Vector Scheduling) In the Vector Scheduling problem, the input is a set of $n$ rational vector jobs $v_1, v_2, \ldots, v_n \in [0,1]^d$, and $m$ identical machines. The objective is to assign the jobs to machines i.e. partition $[n]$ into $m$ parts $A_1, A_2, \ldots, A_m$ to minimize the makespan which is defined as the maximum $\ell_{\infty}$ load on a machine. 
\[
\max_{i \in [m]} \norms{\sum_{j \in A_i}v_j}_{\infty}
\]
\end{definition}

\begin{definition}(Vector Bin Covering)
In the Vector Bin Covering problem, we are given $n$ vectors $v_1, v_2, \ldots, v_n \in [0,1]^d$. The goal is to partition the input vectors into maximum number of parts $A_1, A_2, \ldots, A_m$ such that 
\[
\sum_{j \in A_i}v_j \geq \textbf{1}^d \,\,\forall i \in [m]
\]

\end{definition}

\smallskip \noindent \textbf{Asymptotic Approximation.} 
For the Bin Packing problem, it is NP-Hard to identify if all the vectors can be packed into $2$ bins or need $3$ bins. This already proves that the problem is NP-hard to approximate within $\frac{3}{2}$ as per the usual notion of multiplicative approximation ratio. 
However, this is less interesting as there are much better \textit{asymptotic} approximation algorithms for the problem which get $(1+\epsilon)$-factor approximation when the optimal value is large enough, for every positive constant $\epsilon >0$. 

Even for the Vector Bin Packing problem, the performance of an algorithm is typically measured in the asymptotic setting. 
We give the formal definition~\cite{ChristensenKPT17} of asymptotic approximation ratio of an algorithm $\mathcal{A}$ for the Vector Bin Packing problem.
\begin{definition}(Asymptotic Approximation Ratio)
	The asymptotic approximation ratio $\rho_{\mathcal{A}}^{\infty}$ of an algorithm $\mathcal{A}$ for the Vector Bin Packing problem is 
	\[
	\rho_{\mathcal{A}}^\infty = \limsup\limits_{n \rightarrow \infty} \rho_{\mathcal{A}}^{n},  \,\, \rho_{\mathcal{A}}^n = \sup_{I \in \mathcal{I}} \left\{ \frac{\mathcal{A}(I)}{\emph{OPT}(I)} : \emph{OPT}(I)=n\right\}
	\]
	where $\mathcal{I}$ denotes the set of all possible Vector Bin Packing instances.
\end{definition}
 All the results mentioned in this paper regarding Vector Bin Packing are with respect to the asymptotic approximation ratio.

\smallskip \noindent \textbf{Label Cover.}
We define the Label Cover problem:
\begin{definition}(Label Cover)
In an instance of the Label Cover problem $G=(V=L \cup R, E, \Sigma_L,\Sigma_R,\Pi)$ with $|\Sigma_L| \geq |\Sigma_R|$, the input is a bipartite graph $L \cup R$ with constraints on every edge. The constraint on an edge $e$ is a projection $\Pi_e :\Sigma_L\rightarrow \Sigma_R$. We say a labeling $\sigma : V \rightarrow \Sigma_L \cup \Sigma_R$ satisfies the constraint on the edge $e=(u,v)$ if $\Pi_e(\sigma(u))=\sigma(v)$.
The objective is to find a labeling $\sigma : V \rightarrow \Sigma_L \cup \Sigma_R$ that satisfies as many constraints as possible. 
\end{definition}

By a simple reduction from the $3$-SAT problem, we can prove that Label Cover is NP-hard when $\Sigma_L$ and $\Sigma_R$ are constants (See e.g., Lemma $4.2$ in ~\cite{BG16}). 
\begin{theorem}
\label{thm:lc-hardness}
Given a Label Cover instance when $\Sigma_L=\Sigma_R=[6]$, it is NP-hard to identify if it has a labeling that satisfies all the constraints. 
\end{theorem}

The real use of Label Cover, however, lies in its strong hardness of approximation. PCP Theorem~\cite{ALMSS98} combined with Raz's parallel repetition~\cite{Raz98} yields the following strong inapproximability of Label Cover problem. 
\begin{theorem}
\label{thm:lc-parallel}
There exists an absolute constant $c>1$ such that for every integer $n$ and $\epsilon>0$, there is a reduction from $3$-SAT instance $I$ over $n$ variables to Label Cover instance $G=(V=L \cup R,E,\Sigma_L,\Sigma_R,\Pi)$ with $|V|\leq n^{O\left(\log\left(\frac{1}{\epsilon}\right)\right)}$, $|\Sigma_L|\leq \left(\frac{1}{\epsilon}\right)^{c}$ satisfying the following: 
\begin{enumerate}
    \item (Completeness.) If $I$ is satisfiable, there exists a labeling to $G$ that satisfies all the constraints. 
    \item (Soundness.) If $I$ is not satisfiable, no labeling can satisfy an $\epsilon$ fraction of the constraints of $G$. 
    \item (Biregularity.) The graph $L \cup R, E$ is biregular with degrees on either side bounded by $\emph{\textsf{poly}}\left(\frac{1}{\epsilon}\right)$. 
\end{enumerate}
Furthermore, the running time of the reduction is $\emph{\textsf{poly}}\left(n,\frac{1}{\epsilon}\right)$. 
\end{theorem}

Moshkovitz-Raz~\cite{MR10} proved the following hardness of near linear size Label Cover.
\begin{theorem}
\label{thm:linear-lc}
     There exist absolute constants $c,c'>1$ such that for every $n$ and $\epsilon>0$, there is a reduction from $3$-SAT instance $I$ over $n$ variables to Label Cover instance $G=(V=L \cup R,E,\Sigma_L,\Sigma_R,\Pi)$ with $|V|\leq n^{1+o(1)}\left(\frac{1}{\epsilon}\right)^{c}$, $|\Sigma_L|\leq 2^{\left(\frac{1}{\epsilon}\right)^{c'}}$ satisfying the following: 
\begin{enumerate}
    \item (Completeness.) If $I$ is satisfiable, there exists a labeling to $G$ that satisfies all the constraints. 
    \item (Soundness.) If $I$ is not satisfiable, no labeling can satisfy an $\epsilon$ fraction of the constraints of $G$. 
    \item (Biregularity.) The graph $L \cup R, E$ is biregular with degrees on either side $\emph{\textsf{poly}}\left(\frac{1}{\epsilon}\right)$. 
\end{enumerate}
Furthermore, when $\epsilon$ is a constant, the running time of the reduction is $\emph{\textsf{poly}}(n)$. 
\end{theorem}

\section{Vector Bin Packing}
\label{sec:vbp}
In this section, we prove the hardness of approximation of Vector Bin Packing. First, we define the packing dimension of a set family and bound the packing dimension of simple set families. Next, we combine this upper bound with the hardness of set cover on simple bounded set systems to prove~\Cref{thm:main-vbp}.

\subsection{Packing Dimension}

For a set family $\mathcal{S}$ on a universe $V$, we define the packing dimension $\textsf{pdim}(\mathcal{S})$ below. For a function $f: V \rightarrow [0,1]^K$ and a set $S \subseteq V$, we let $f(S)$ denote the vector $f(S)=\sum_{v \in S}f(v)$. 
\begin{definition}
For a set family $\mathcal{S}$ on a universe $V$, the packing dimension $\emph{\textsf{pdim}}(\mathcal{S})$ is defined as the smallest positive integer $K$ such that there exists an embedding $f:V\rightarrow [0,1]^K$ that satisfies the following property: For every set $S \subseteq V$, $S$ is in the family $\mathcal{S}$ if and only if 
\[
\norm{ f(S) }_{\infty} \leq 1.
\]
If no such embedding exists, we say that $\emph{\textsf{pdim}}(\mathcal{S})$ is infinite. 
\end{definition}
For a set family $\mathcal{S}$ to have finite packing dimension i.e. for an embedding $f : V \rightarrow [0,1]^K$ realizing the above condition to exist requires two conditions:
\begin{enumerate}
    \item The set family is \textit{downward closed} i.e. for every $S \in \mathcal{S}$ and $T\subseteq S$, $T\in \mathcal{S}$ as well.
    \item For every element $v \in V$, there is a set $S \in \mathcal{S}$ with $v \in S$. We call a set family $\mathcal{S}$ on a universe $V$ non-trivial if for every $v \in V$, there is a set $S \in \mathcal{S}$ with $v \in S$. 
\end{enumerate}
On the other hand, any set system that satisfies the above two conditions i.e. being downward closed and non-trivial has a finite packing dimension. Before proving this statement, we first prove the following simple but useful proposition. 
\begin{proposition}
\label{prop:dim-intersection}
For a pair of set families $\mathcal{S}_1$ and $\mathcal{S}_2$ defined on the same universe $V$ such that $\emph{\textsf{pdim}}(\mathcal{S}_1)$ and $\emph{\textsf{pdim}}(\mathcal{S}_2)$ are finite, 
\[
\emph{\textsf{pdim}}(\mathcal{S}_1 \cap \mathcal{S}_2) \leq \emph{\textsf{pdim}}(\mathcal{S}_1) + \emph{\textsf{pdim}}(\mathcal{S}_2) 
\]
\end{proposition}
\begin{proof}
	Let $K_1 = \textsf{pdim}(\mathcal{S}_1)$ and $K_2 = \textsf{pdim}(\mathcal{S}_2)$. 
	Suppose that $f_1 : V \rightarrow [0,1]^{K_1}$ be such that for every set $S \subseteq V$,
	\[
	\norm{f_1(S)}_{\infty} \leq 1
	\]
	if and only if $S \in \mathcal{S}_1$. Similarly, let $f_2 : V \rightarrow [0,1]^{K_2}$ be such that for every set $S \subseteq V$, 
	\[
	\norm{f_2(S)}_{\infty}\leq 1
	\]
	if and only if $S \in \mathcal{S}_2$. 
	Consider the function $f : V \rightarrow [0,1]^{K_1 + K_2}$ defined as $f(v)=(f_1(v),f_2(v))$. Then, for every set $S \subseteq V$, $\norm{f(S)}_{\infty}\leq 1$ if and only if $\norm{f_1(S)}_{\infty}\leq 1$ and $\norm{f_2(S)}_{\infty}\leq 1$. Thus, for every set $S \subseteq V$, $\norm{f(S)}_{\infty}\leq 1$ if and only if $S \in \mathcal{S}_1$ and $S \in \mathcal{S}_2$, or equivalently, if $S \in \mathcal{S}_1 \cap \mathcal{S}_2$. Hence, the packing dimension of $\mathcal{S}_1 \cap \mathcal{S}_2$ is at most $K_1 + K_2$. 
\end{proof}

For a set $S \subseteq V$, let $S^{\uparrow}$ be the family of sets $T \subseteq V$ such that $S \subseteq T$. Similarly, let $S^{\downarrow}$ be the family of sets $T \subseteq V$ such that $T \subseteq S$. For a set system $\mathcal{S}$, we let $\mathcal{S}^{\uparrow}$ (resp. $\mathcal{S}^{\downarrow}$) denote the union of $S^{\uparrow}$ (resp. $S^{\downarrow}$) over all $S \in \mathcal{S}$.

Consider a set $S \subseteq V$ with $|S|>1$. For the set family $2^V\setminus S^\uparrow$, we have the embedding $f:V \rightarrow [0,1]$ defined as 
\[
f(v) = \begin{cases}
\frac{1}{|S|}+\frac{1}{|S|^2}, \text{ if }v\in S \\ 
0 \text{ otherwise. }
\end{cases}
\]
This shows that $\textsf{pdim}(2^V \setminus S^{\uparrow})\leq 1$ for all $S\subseteq V$ with $|S|>1$. Note that we have 
\[
\mathcal{S} = \bigcap_{S \notin \mathcal{S}} 2^V \setminus S^{\uparrow}
\]
for every downward closed set system $\mathcal{S}$.
Combined with~\Cref{prop:dim-intersection}, we obtain that for every non-trivial downward closed family $\mathcal{S}$ on a universe $V$, $\textsf{pdim}(\mathcal{S})\leq 2^{|V|}$.

We are interested in the classes of set families for which there is an efficient embedding with packing dimension being independent of $|V|$. 
In particular, the class of set families that we study are bounded set families where each set has cardinality at most $k$, and each element appears in at most $\Delta$ sets.
We can show that such bounded set families that are downward closed and non-trivial have packing dimension at most $(k\Delta)^{O(\Delta)}$. Together with the $\Omega(\log k)$ hardness~\cite{Trevisan01} of $k$-set cover where each set has cardinality at most $k$, a fixed constant and each element appearing in $(\log k)^{O(1)}$ sets, this packing dimension bound gives the hardness of $(\log d)^{\Omega(1)}$ for the Vector Bin Packing problem when $d$ is a large constant. 
Unfortunately, the exponential dependence on $\Delta$ is necessary for the packing dimension of bounded set systems, and thus, this approach does not yield the optimal $\Omega(\log d)$ hardness of Vector Bin Packing. 

Instead of using arbitrary bounded set families, we bypass this barrier by using simple bounded set families. 
Recall that a set family is called simple if any two distinct sets in the family intersect in at most one element. 
It turns out that for simple bounded set families i.e. simple set families $\mathcal{S}$ where each set has cardinality at most $k$, and each element appears in at most $\Delta$ sets, the packing dimension of $\mathcal{S}^\downarrow$ can be upper bounded by $(k\Delta)^{O(1)}$. Together with the $\Omega(\log k)$ hardness of simple $k$-set cover (proved in~\Cref{sec:simple}), we get the optimal $\Omega(\log d)$ hardness of Vector Bin Packing when $d$ is a large constant. 
In the next subsection, we prove the packing dimension upper bound, and we use this upper bound to prove the hardness of Vector Bin Packing in~\Cref{subsec:vbp-hardness}. 

\subsection{Packing Dimension of Simple Bounded Set Families}
The main embedding result that we prove is that the downward closure of simple set systems where each set has cardinality $k$ and each element appears in at most $\Delta$ sets has packing dimension at most polynomial in $k,\Delta$. 
\begin{theorem}
	\label{thm:bounded-dimension}
	Suppose that $\mathcal{S}$ is a simple non-trivial set system on a universe $V$ where each set has cardinality at most $k \geq 2$ and each element appears in at most $\Delta$ sets. Then, 
	\[
	\emph{\textsf{pdim}}(\mathcal{S}^{\downarrow}) \leq (k\Delta)^{O(1)}
	\]
	Furthermore, an embedding realizing the above can be found in time polynomial in $|V|$.
\end{theorem}

We prove the embedding result by writing the set family $\mathcal{S^\downarrow}$ as an intersection of $(k\Delta)^{O(1)}$ structured set families each of which has packing dimension at most $(k\Delta)^{O(1)}$. We can then upper bound the packing dimension of $\mathcal{S}^\downarrow$ using~\Cref{prop:dim-intersection}. 
The structured set systems we study are \textit{sunflower-bouquets}, which are a disjoint union of sunflowers\footnote{A sunflower is a collection of sets $S_1, S_2, \ldots, S_r$ whose pairwise intersection is constant i.e., there exists $U$ such that $U=S_i \cap S_j$ for all $i,j \in [r], i \neq j$. This constant intersection $U$ is called the kernel of the sunflower.} that have a single element as the kernel. The formal definition of the sunflower-bouquet set families is below. See~\Cref{fig:simple-set} for an illustration. 
\begin{definition}(Sunflower-bouquets)
A simple set system $\mathcal{S}$ on a universe $V$ is called a sunflower-bouquet with core $U \subseteq V, U \neq \phi$ if the following hold. 
\begin{enumerate}
    \item Every set $S \in \mathcal{S}$ satisfies $|S \cap U|=1$. Furthermore, for every $u \in U$, there is a set $S \in \mathcal{S}$ with $u \in S$. 
    \item For any pair of sets $S_1, S_2 \in \mathcal{S}$ with $S_1 \cap S_2 \neq \emptyset$, we have $S_1 \cap U = S_2 \cap U = S_1 \cap S_2$.
\end{enumerate}
\end{definition}

\smallskip 

We now give an efficient embedding for a sunflower-bouquet $\mathcal{S}$ on a universe $V$ with core $U \subseteq V, U \neq \phi$. 
The motivation behind this lemma is to upper bound the packing dimension of the set system $\mathcal{T^\downarrow} = \mathcal{S}^\downarrow \cup \{ S \subseteq V \setminus U : |S| \leq k \}$. 

\begin{lemma}
\label{lem:structured}
Fix an integer $k \geq 2$. Let $\mathcal{S}$ be a simple set family defined on a universe $V$ that is a sunflower-bouquet with core $U$. Furthermore, each set in the family has cardinality at most $k$ and each element appears in at most $\Delta$ sets. 
Then, there exists an embedding $f: V \rightarrow [0,1]^K$ that satisfies
\begin{enumerate}
    \item[(A)] For every set $S \in \mathcal{S}$, 
    \[
    \norm{f(S)}_{\infty} \leq 1.
    \]
    \item[(B)] For every set $S \notin \mathcal{S}^{\downarrow}$ with $S \cap U \neq \emptyset$, 
    \[
    \norm{f(S)}_{\infty}>1.
    \]
    \item[(C)] For every set $S \subseteq V$ with $S \cap U = \emptyset$ and $|S|\leq k$, 
    \[
    \norm{f(S)}_{\infty} \leq 1.
    \]
    \item[(D)] For every set $S\subseteq V$ with $|S|>k$,
    \[
    \norm{f(S)}_{\infty}>1.
    \]
\end{enumerate}
with $K=(k\Delta)^{O(1)}$. Furthermore, such an embedding can be found in time polynomial in $|V|$ given $\mathcal{S}$.
\end{lemma}

\begin{proof}
Let $U=\{u_1, u_2, \ldots, u_m\}$. We can partition $V\setminus U$ into $V_0, V_1, \ldots, V_m$ with
\[
V_i = \bigcup_{S \in \mathcal{S} : u_i \in S} S \setminus \{ u_i \} 
\]
for all $i \in [m]$. 
Here, $\mathcal{S}$ restricted to $\{u_i\} \cup V_i$ is a sunflower set system with a single element $u_i$ as the kernel for every $i \in [m]$.
%Thus, for every $i \in [m]$, every set $S \in \mathcal{S}$ with $u_i \in S$ has $S \setminus \{u_i\} \subseteq V_i$. 
%We also observe that for every $i \in [m]$ and $v \in V_i$, there exists exactly one set $S \in \mathcal{S}$ with $v \in S$ and that set $S$ satisfies $u_i \in S$. 
As each set in $\mathcal{S}$ has cardinality at most $k$ and each element appears in at most $\Delta$ sets, we get that $|V_i|\leq k\Delta$ for all $i \in [m]$. For every $i \in [m]$, we order the elements of $V_i$ as $ \{ v_{i,1}, v_{i,2}, \ldots, v_{i,k\Delta}\}$ (with repetitions if needed).

\begin{figure}
	\includegraphics{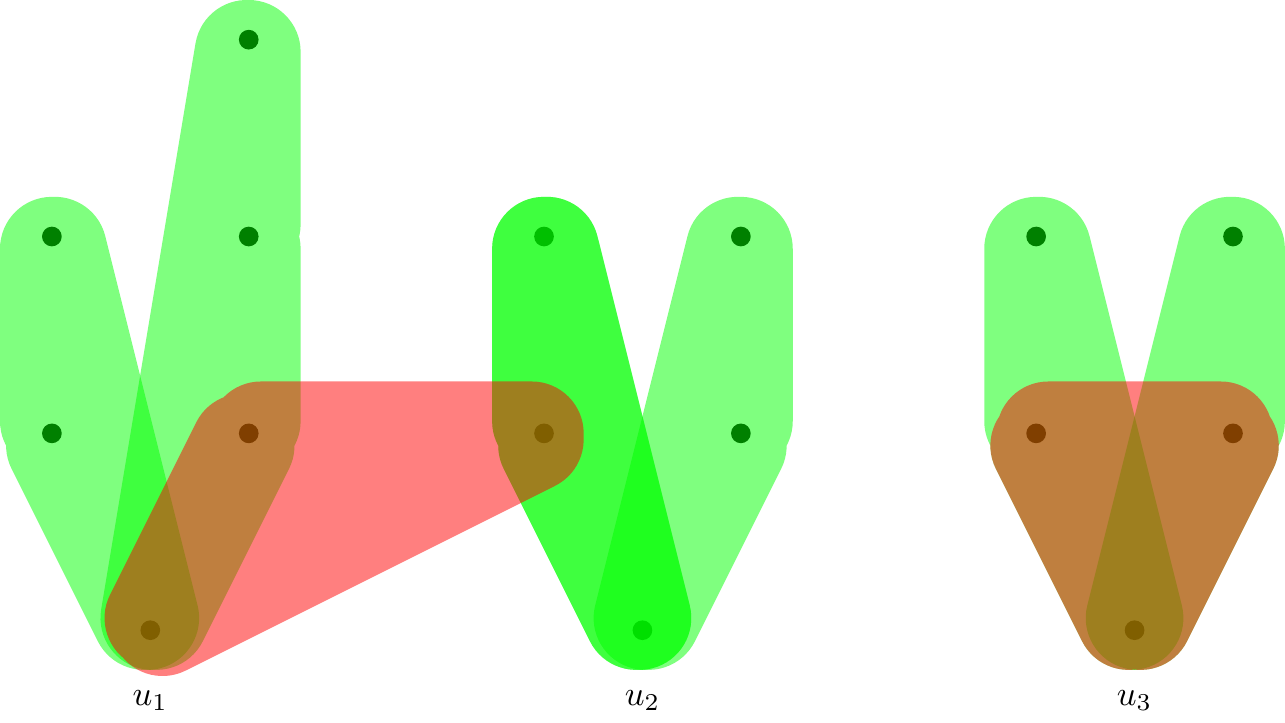}
	\centering
	\caption{An illustration of a sunflower-bouquet set family. Here, $\mathcal{S}$ is the family of all the green colored sets. It is a sunflower-bouquet with core $U=\{u_1, u_2, u_3\}$. 
		In the embedding, we ensure that the $\ell_{\infty}$ norm of the left red set is greater than $1$ in the first step while the right side red set is handled in the second step.}
	\label{fig:simple-set}
\end{figure}

We construct the final embedding $f$ as a concatenation of embeddings with smaller dimensions $f := (f_0, g, g')$ where $f_0 : V \rightarrow [0,1]^2$, $g: V \rightarrow [0,1]^{K_1}$ and $g' : V \rightarrow [0,1]^{K_2}$ all satisfy the conditions $(A)$ and $(C)$. In other words, for every set $S \subseteq V$ satisfying either $S \in \mathcal{S}$ or $S \cap U = \phi$ and $|S| \leq k$, we have $\norm{f_0(S)}_{\infty} \leq 1$, $\norm{g(S)}_{\infty}\leq 1$, and $\norm{g'(S)}_{\infty}\leq 1$. 
Note that for a set $S \subseteq V$, we have 
\[
\norm{f(S)}_{\infty} = \max (\norm{f_0(S)}_{\infty}, \norm{g(S)}_{\infty}, \norm{g'(S)}_{\infty})
\]
Thus, if $f_0, g, g'$ satisfy the conditions $(A)$ and $(C)$, the final embedding $f$ also satisfies the conditions $(A)$ and $(C)$.
Furthermore, the parameters $K_1$ and $K_2$ are chosen such that the final dimension of $f$, $2+K_1+K_2$ is at most $(k\Delta)^{O(1)}$. 

% Note that for a set $S \subseteq V$, we have 
% \[
% \norm{f(S)}_{\infty} = \max (\norm{f_0(S)}_{\infty}, \norm{g(S)}_{\infty}, \norm{g'(S)}_{\infty})
% \]
% Thus, to ensure that the final embedding $f$ satisfies the conditions $(B)$ and $(D)$, we need that for every $S \subseteq V$ with either $|S|>k$, or $S \notin \mathcal{S}^\downarrow$ and $S \cap U \neq \phi$, at least one of $\norm{f_0(S)}_{\infty}, \norm{g(S)}_{\infty}, \norm{g'(S)}_{\infty}$ is strictly larger than $1$. 

First, we define the embedding $f_0 : V\rightarrow [0,1]^2$ that satisfies the conditions $(A)$ and $(C)$ with the additional property that for every set $S \subseteq V$ with $|S| > k$, or if $S \cap U \neq \phi$ and $S \cap V_0 \neq \phi$, or if $|S \cap U|>1$, we have $\norm{f_0(S)}_{\infty}>1$. 
We obtain this by a simple two-dimensional embedding as follows:
\[
f_0(v)=
\begin{cases}
\left(1,\frac 1k\right), \text{ if } v \in U \\
\left(\frac 1k,\frac 1k\right), \text{ if } v \in V_0\\
\left(0,\frac 1k\right) \text{ otherwise.} 
\end{cases}
\]
We can verify that $f_0$ satisfies the conditions $(A)$ and $(C)$. Furthermore, suppose that $\norm{f_0(S)}_{\infty} \leq 1 $ for a set $S\subseteq V$. Then, we can obtain the following observations that we will use later. 
\begin{enumerate}
    \item As $f_0(v)_1 = 1 $ for all $v \in U$, $|S \cap U| \leq 1$. As $f_0(v)_1 = \frac 1k$ for all $v \in V_0$, if $S \cap U \neq \emptyset$, then $S \cap V_0 = \emptyset$.
    \item As $f_0(v)_2 = \frac 1k$ for all $v \in V$, $|S|\leq k$. Thus, for every set $S \subseteq V$ such that $|S|>k$, we have $\norm{f_0(S)}_{\infty}>1$, and hence, $\norm{f(S)}_{\infty}>1$. This already proves the condition $(D)$ of the lemma. 
\end{enumerate}

\smallskip \noindent \textbf{Overview of rest of the proof.} We now restrict our attention to sets $S \subseteq V$ such that $|S|\leq k, S \cap V_0 = \phi$, and $|S \cap U| \leq 1$. 
Our goal is to find an embedding for these sets that satisfies the conditions $(A)$, $(B)$, and $(C)$.
This is the technically challenging part of the proof and requires setting various coordinates carefully to encode the properties of the set system. 
We break this down into two steps: eliminating the ``cross-sunflower'' sets, and pinning down the ``intra-sunflower'' sets. We give an overview of the ideas used in the two steps before presenting the full proof.  
\begin{enumerate}
    \item The cross-sunflower sets are the sets $S \subseteq V$ that contain $u_i$ for some $i \in [m]$, but also intersect another sunflower i.e., $S \cap V_{i'} \neq \phi$ for an $i' \neq i$. Note that such a set $S$ satisfies $S \notin \mathcal{S}^{\downarrow}$ and $S \cap U \neq \phi$, and thus, to satisfy the condition $(B)$, we need to ensure that $\norm{f(S)}_{\infty}>1$ for such sets. We achieve this by constructing an embedding $g:V\rightarrow [0,1]^{K_1}$ that satisfies the conditions $(A)$ and $(C)$ and has $\norm{g(S)}_{\infty}>1$ for all cross-sunflower sets.
    
    We illustrate the idea used in constructing this embedding using a toy example. Suppose that we have a set of pairs of elements $(u_1, v_1), (u_2, v_2), \ldots, (u_n,v_n)$, and let their union be denoted by $W = \{u_1, v_1, u_2, v_2, \ldots, u_n, v_n \}$. Our goal is to find an embedding $g : W \rightarrow [0,1]^2$ such that 
    \begin{enumerate}
        \item $\norm{g(u_i)+g(v_i)}_{\infty}\leq 1$ for all $i \in [n]$. 
        \item $\norm{g(u_i)+g(v_j)}_{\infty}>1$ for all $i, j \in [n], i \neq j$. 
    \end{enumerate}
    We construct this embedding by choosing $n$ distinct real numbers $\alpha_1, \alpha_2, \ldots, \alpha_n \in (0,1)$ and setting $g(u_i)=(\alpha_i,1-\alpha_i)$ and $g(v_i)=(1-\alpha_i,\alpha_i)$ for all $i \in [m]$. Note that $\norm{g(u_i)+g(v_j)}_1=2$ for all $i,j$, and thus, $\norm{g(u_i)+g(v_j)}_{\infty}=1$ if and only if $g(u_i)+g(v_j)=(1,1)$, or equivalently, when $i=j$. The actual construction extends this idea with two differences: first, the $V_i$s could have more than one element, and we use the pairs idea multiple times to account for this, and second, we need to ensure that the sum of the embedding of any $k$ elements in $V \setminus (U \cup V_0)$ has $\ell_{\infty}$ norm at most $1$, and thus, we need to choose the embedding of $u_i$ to be $(\alpha_i, 2-\frac 1k - \alpha_i)$ and set $\alpha_i \in (1-\frac 1k, 1)$. 
    \item The intra-sunflower sets are the sets $S \subseteq V$ such that $u_i \in S$ and $S  \subseteq V_i\cup \{ u_i \}$ for some $i \in [m]$. We need to ensure that every intra-sunflower set $S$ such that $S \notin \mathcal{S}^\downarrow$ satisfies $\norm{f(S)}_{\infty}>1$. We achieve this by constructing an embedding $g' : V \rightarrow [0,1]^{K_2}$ that satisfies the conditions $(A)$, $(C)$ and has $\norm{g'(S)}_{\infty}>1$ for every intra-sunflower set $S$ such that $S \notin \mathcal{S}^\downarrow$.
    
    Fix an $i \in [m]$. For every intra-sunflower set $S \subseteq \{u_i \} \cup V_i$ with $u_i \in S$ and $S \notin \mathcal{S}^\downarrow$, we use a single dimensional embedding $g_S : V \rightarrow [0,1]$ that satisfies the conditions $(A)$ and $(C)$ but has $\norm{g_S(S)}_{\infty}>1$ . We achieve this by setting $g_S(u_i)=1-\frac{|S|-1}{k}+\epsilon$, and $g_S(v)=\frac{1}{k}$ for every $v \in V_i \cap S$, and $g_S(v)=0$ for all $v \in V \setminus S$, where $ \epsilon < \frac 1k$ is a small positive constant. Note that $\norm{g_S(T)}\leq 1$ for every set $T$ such that $T\nsupseteq S$, and since $g_S(v) \leq \frac 1k$ for all $v \in V \setminus U$, the embedding $g_S$ satisfies the conditions $(A)$ and $(C)$.
    
    We can construct such a single dimensional embedding for every intra-sunflower set $S \notin \mathcal{S}^\downarrow$ and take their concatenation to obtain the required embedding $g'$. However, there could be exponential (in $k,\Delta$) number of such intra-sunflower sets $S \subseteq \{u_i\} \cup V_i$, $u_i \in S$ such that $S \notin \mathcal{S}^\downarrow$. We get around this issue by observing that we need the single dimensional embedding only for \textit{minimal} intra-sunflower sets that don't belong to $\mathcal{S}^\downarrow$. In fact, as the set system $\mathcal{S}^\downarrow$ restricted to $\{u_i\}\cup V_i$ is a sunflower with single element $u_i$ as the kernel, we can deduce that every minimal intra-sunflower set $S$ with $S \notin \mathcal{S}^\downarrow$ is of the form $\{u_i, x, y\}$ where $x,y \in V_i$. Thus, we can construct the single dimensional embedding for such sets and take their concatenation to obtain the required embedding $g'$ with dimension at most $|V_i|^2 \leq (k\Delta)^2$. 
%TODO this needs some editing, and the notation needs some fixes.     
\end{enumerate}
As every set $S \subseteq V$ such that $S \cap V_0 = \phi$, $|S \cap U| = 1$ is either a cross-sunflower set or an intra-sunflower set, the two steps together prove that $f = (f_0,g,g')$ satisfies the conditions $(A)$, $(B)$ and $(C)$. 

We now present the full formal proof of the two steps. 

\smallskip \noindent \textbf{Step-1. Eliminating the cross-sunflower sets}. In the first step, our goal is to find an embedding $g : V \rightarrow [0,1]^{K_1}$ with $K_1=2k\Delta $ such that 
\begin{enumerate}
    \item $g$ satisfies the conditions $(A)$ and $(C)$ i.e. for every set $S \in \mathcal{S}$, $\norm{g(S)}_{\infty}\leq 1$, and for every set $S \subseteq V$ with $S \cap U = \phi$ and $|S| \leq k, \norm{g(S)}_{\infty}\leq 1.$
    \item For every ``cross-sunflower'' set $S \subseteq V$ with $u_i \in S$ for some $i \in [m]$, and $S \cap V_{i'} \neq \phi$ for $i' \in [m], i' \neq i$, we have $\norm{g(S)}_{\infty} >1$.
\end{enumerate}
We achieve this by setting $g=(f_1,\ldots,f_{k\Delta})$, where each $f_l : V \rightarrow [0,1]^2, l \in [k\Delta]$ satisfies the conditions $(A)$ and $(C)$, and overall, the embedding $g$ satisfies the second condition above. 
\smallskip

We choose $m$ distinct rational numbers $ \alpha_1, \ldots, \alpha_m$ with $1-\frac 1k < \alpha_i < 1$ for all $i \in [m]$. 
We define the embeddings $f_l: V \rightarrow [0,1]^2$, $l \in [k\Delta]$ as follows.
Consider an $l \in [k\Delta]$. 
\begin{enumerate}
\item For $i \in [m]$, we set 
\[
f_l(u_i) = \left( \alpha_i, 2-\frac 1k - \alpha_i \right)
\]
\item For $i \in [m]$ and $v_{i,j} \in V_i$, we set $f_l(v_{i,j})=(0,0)$ if $v_{i,j} \neq v_{i,l}$. We set 
\[
f_l(v_{i,l})=\left( 1 - \alpha_i, \alpha_i + \frac 1k - 1 \right)
\]
\item For $v \in V_0$, we set $f_l(v)=(0,0)$.
\end{enumerate}
We verify that these embeddings satisfy the conditions $(A)$ and $(C)$. Fix an $l \in [k\Delta]$.

\begin{enumerate}
    \item[(A)] Consider a set $S \in \mathcal{S}$. Let $i \in [m]$ be such that $\{u_i\} = S \cap U$.
    We have 
    \begin{align*}
    f_l(S) &= \sum_{v \in S}f_l(v) \\ 
    &\leq \sum_{v \in \{u_i\}\cup V_i} f(v) \\ 
    &= f_l(u_i)+f_l(v_{i,l}) \\
    &= \left( \alpha_i, 2-\frac 1k - \alpha_i \right) + \left( 1 - \alpha_i, \alpha_i + \frac 1k - 1 \right) = (1,1).
    \end{align*}
    \item[(C)] This follows directly from the fact that $\norm{f_l(v)}_1\leq \frac 1k$ for all $l \in [k\Delta]$ and $v \in V \setminus U$. 
\end{enumerate}
Let $g:V \rightarrow [0,1]^{2k\Delta}$ be defined as $g=(f_1,\ldots, f_{k\Delta})$. As each of the individual embeddings satisfies $(A)$ and $(C)$, $g$ also satisfies the conditions $(A)$ and $(C)$. 

Let $S \subseteq V \setminus V_0$, $|S \cap U |=1$ be such that 
\[
\norm{g(S)}_{\infty} \leq 1.
\]
i.e. $\norm{f_l(S)}_{\infty}\leq 1$ for all $l \in [k\Delta]$.
Suppose that $S \cap U = \{u_i\}$. Then, we claim that $S \subseteq \{u_i\} \cup V_i$. Suppose for contradiction that this is not the case, and there exists $v_{i',l} \in V_{i'}$ with $i' \neq i, i' \in [m]$ and $l \in [k\Delta]$ such that $v_{i',l} \in S$. We have 
    \begin{align*}
        f_l(S) &= \sum_{v \in S}f_l(v) \\ 
        &\geq f_l(u_i) + f_l(v_{i',l})\\ 
        &= \left( \alpha_i, 2-\frac 1k - \alpha_i \right) + \left( 1-\alpha_{i'}, \alpha_{i'}+\frac 1k - 1 \right) \\ 
        &= ( 1+ \alpha_i - \alpha_{i'}, 1 + \alpha_{i'}-\alpha_i )
    \end{align*}
    As $\alpha_i \neq \alpha_{i'}$, $\norm{f_l(S)}_{\infty}>1$, a contradiction.
Thus, for every set $S \subseteq V$ such that $u_i \in S$, $S \cap V_{i'}\neq \phi$ for some $i' \neq i$, we have $\norm{g(S)}_{\infty}>1$. 

\smallskip \noindent \textbf{Step 2. Pinning down the intra-sunflower sets.} In the second step, our goal is to find an embedding $g' : V \rightarrow [0,1]^{K_2}$ with $K_2=(k\Delta)^2$ such that 
\begin{enumerate}
    \item $g'$ satisfies the conditions $(A)$ and $(C)$. 
    \item For every $i \in [m]$ and ``intra-sunflower'' set $S \subseteq \{u_i\} \cup V_i$ such that $u_i \in S$ and $S \notin \mathcal{S}^\downarrow$, we have $\norm{g'(S)}_{\infty}>1$.
\end{enumerate}
We achieve this by setting $g'=(g_1, g_2, \ldots, g_{(k\Delta)^2)})$ where each $g_l$, $l \in [(k\Delta)^2]$ satisfies the conditions $(A)$ and $(C)$, and the overall function $g'$ satisfies the second condition above. 

\smallskip 

For every $i \in [m]$, we order all the pairs of distinct elements  $x,y \in V_i$ as $\{ V_{i,1}, V_{i,2}, \ldots, V_{i,(k\Delta)^2}\}$ (with repetitions if needed). The upper bound on the number of such pairs is obtained using the fact that $|V_i|\leq k\Delta$ for all $i \in [m]$.

We define the embeddings $g_l:V \rightarrow [0,1]$, $l \in [(k\Delta)^2]$ below. Fix an $l \in [(k\Delta)^2]$. 
\begin{enumerate}
    \item Consider an $i \in [m]$. We have two different cases:
    \begin{enumerate} 
    \item If $V_{i,l} \cup \{u_i\} \in \mathcal{S}^{\downarrow}$, we set $g_l(u_i)=0$ and $g_l(v)=0$ for all $v \in V_i$.
    \item If $V_{i,l} \cup \{u_i\} \notin \mathcal{S}^{\downarrow}$, we set $g_l(v)=\frac{1}{k}$ for all $v \in V_{i,l}$, and $g_l(v)=0$ for all $v \in V_i \setminus V_{i,l}$. We set 
    \[
    g_l(u_i) = 1 - \frac{2}{k} + \frac{1}{k^2}
    \]
    \end{enumerate}
    \item For all $v \in V_0$, we set $g_l(v)=0$. 
\end{enumerate}

We now verify that these embeddings satisfy the conditions $(A)$ and $(C)$. Fix an integer $l \in [(k\Delta)^2]$.
\begin{enumerate}
    \item[(A)] Consider a set $S \in \mathcal{S}$. Let $\{u_i\}=S \cap U$. If $\{u_i\} \cup V_{i,l} \in \mathcal{S}^{\downarrow}$, $g_l(v)=0$ for all $v \in S$, and thus we have $|g_l(S)|\leq 1$. 
    Now suppose that $\{u_i\} \cup V_{i,l} \notin \mathcal{S}^{\downarrow}$. This implies that $V_{i,l}$ is not a subset of $S$. As $|V_{i,l}|=2$, $|V_{i,l} \cap S|\leq 1$. We get 
    \begin{align*}
        \sum_{v \in S}g_l(v) &= g_l(u_i)+\sum_{v \in S \cap V_i} g_l(v) \\ 
        &= g_l(u_i) + \sum_{v \in S \cap V_{i,l}} g_l(v) \\ 
        &\leq g_l(u_i) + \frac{1}{k} \\
        &= 1- \frac{2}{k}+\frac{1}{k^2} + \frac{1}{k} \leq 1
    \end{align*}
    \item[(C)] This follows from the fact that $g_l(v)\leq \frac 1k$ for all $v \in V \setminus U$.
\end{enumerate}

Suppose that a set $S\subseteq V$ satisfies $S \subseteq \{u_i\} \cup V_i$ for some $i \in [m]$, and $u_i \in S$, $S \notin \mathcal{S}^\downarrow$. Then, we claim that $\norm{g'(S)}_{\infty}>1$. Suppose for the sake of contradiction that $\norm{g'(S)}_{\infty}\leq 1$. Then, we have $g_l(S)\leq 1$ for all $l \in [(k\Delta)^2]$.
Let $S = \{ u_i, s_1, s_2, \ldots, s_p \}$ where $s_j\in V_i$ for all $j \in [p]$. Note that for every $v \in V_i$, there is exactly one set $S(v) \in \mathcal{S}$ such that $v \in S(v)$ and this set $S(v)$ satisfies $u_i \in S(v)$. This follows from the definition of $V_i$ and the fact that the set family $\mathcal{S}$ is a sunflower-bouquet.

We now claim that $S(s_{j_1})=S(s_{j_2})$ for all $j_1, j_2 \in [p]$.
Suppose for contradiction that there exist $j_1, j_2 \in [p]$ with $S(s_{j_1}) \neq S(s_{j_2})$.
This implies that $\{u_i, s_{j_1}, s_{j_2}\} \notin \mathcal{S}^{\downarrow}$ as otherwise, if there exists $T \in \mathcal{S}$ such that $\{ u_i, s_{j_1}, s_{j_2} \} \subseteq T$, we have $S(s_{j_1})=S(s_{j_2})=T$.
Let $l \in [(k\Delta)^2]$ be such that $V_{i,l} = \{ s_{j_1}, s_{j_2}\}$. As  $V_{i,l} \cup \{ u_i\} \notin \mathcal{S}^{\downarrow}$, we have $g_l(v)=\frac{1}{k}$ for all $v \in V_{i,l}$ and 
\[
g_l(u_i)=1-\frac{2}{k}+\frac{1}{k^2} 
\]
Thus, we get that 
\begin{align*}
\sum_{v \in S}g_l(v) &= g_l(u_i) + \sum_{v \in S \setminus \{u_i\}} g_l(v) \\ 
                    &= g_l(u_i) + \sum_{v \in V_{i,l}} g_l(v)\\
                    &=1 - \frac{2}{k}+\frac{1}{k^2} + \frac{2}{k} = 1 + \frac{1}{k^2}
\end{align*}
contradicting the fact that $g_l(S)\leq 1$. 
This completes the proof that $S(s_{j_1})=S(s_{j_2})$ for all $j_1, j_2 \in [p]$. 
Thus, there exists  a set $S(s_1) \in \mathcal{S}$ such that $S \subseteq S(s_1)$, which implies that $S \in \mathcal{S}^{\downarrow}$, a contradiction. Thus, for every set $S \subseteq V$ such that $u_i \in S$, $S \subseteq \{u_i\} \cup V_i$ for some $i \in [m]$ and $\norm{g'(S)}_{\infty} \leq 1$, we have 
$S \in \mathcal{S}^\downarrow$.

\smallskip \noindent \textbf{Final embedding.}
We define the final embedding $f:V \rightarrow [0,1]^{2+2k\Delta + (k\Delta)^2}$ as $f=(f_0,g,g')$. As each of these embeddings satisfies the conditions $(A)$ and $(C)$, the final embedding $f$ also satisfies the conditions $(A)$ and $(C)$.

Suppose that $\norm{f(S)}_{\infty}\leq 1$ for a set $S \subseteq V$. Then, $\norm{f_0(S)}_{\infty}\leq 1$, $\norm{g(S)}_{\infty}\leq 1$ and $\norm{g'(S)}_{\infty}\leq 1$. Condition $(D)$ follows immediately as $\norm{f_0(S)}_{\infty}\leq 1$ implies that $|S|\leq k$.

We now return to condition $(B)$. 
Suppose that $S \subseteq V$ with $S \cap U \neq \emptyset$ satisfies $\norm{f(S)}_{\infty} \leq 1$. Our goal is to show that $S \in \mathcal{S}^{\downarrow}$. 
We have already deduced from $\norm{f_0(S)}_{\infty}\leq 1$ that $|S \cap U|\leq 1$. As $S \cap U \neq \phi$, we have $|S \cap U|=1$, and by using $\norm{f_0(S)}_{\infty}\leq 1$ again, we get that $S \cap V_0 = \phi$. 
Let $S \cap U =\{u_i\}$. As $\norm{g(S)}_{\infty}\leq 1$, using the argument in the first step, we can conclude that $S \cap V_{i'} = \phi$ for all $i' \neq i$. Thus, $S \subseteq \{u_i\} \cup V_i$. By using the argument in the second step, $\norm{g'(S)}_{\infty}\leq 1$ implies that $S \in \mathcal{S}^\downarrow$.

Note that our construction is explicit, and we have a polynomial time algorithm to output the required embedding. 
The dimension of the embedding is $2+2k\Delta + (k\Delta)^2$, which is at most $(k\Delta)^{O(1)}$.
\end{proof}

As a corollary, we bound the packing dimension of the set family 
\[
\mathcal{T}^\downarrow = \mathcal{S}^\downarrow \cup \{ S \subseteq V \setminus U : |S|\leq k\}.
\]
\begin{corollary}
\label{cor:dimension}
Suppose that $\mathcal{T}$ is a set family defined on a universe $V$ with 
\[
\mathcal{T} = \mathcal{S} \cup \{ S \subseteq V \setminus U : |S| \leq k \}
\]
where $\mathcal{S} \subseteq 2^V$ is a sunflower-bouquet with core $U$. Furthermore, each set in $\mathcal{S}$ has cardinality at most $k\geq 2$ and each element appears in at most $\Delta$ sets in $\mathcal{S}$.
Then, 
\[
\emph{\textsf{pdim}}(\mathcal{T}^{\downarrow})\leq (k\Delta)^{O(1)}
\]
Furthermore, an embedding realizing this packing dimension can be found in time polynomial in $|V|$ given $\mathcal{S}$.
\end{corollary}
\begin{proof}
As $\mathcal{S}$ is a sunflower-bouquet, from ~\Cref{lem:structured}, there exists an embedding $f:V\rightarrow [0,1]^K$ that satisfies the conditions $(A),(B),(C)$ and $(D)$ with $K=(k\Delta)^{O(1)}$. Conditions $(A)$ and $(C)$ together imply that 
\[
\norm{f(S)}_{\infty} \leq 1
\]
for all $S \in \mathcal{T}$. 
Note that 
\[
\mathcal{T}^{\downarrow} = \mathcal{S}^{\downarrow} \cup \{ S \subseteq V \setminus U : |S| \leq k \}.
\]
Suppose that $S \subseteq V$ is a subset of $V$ with $S \notin \mathcal{T}^{\downarrow}$. If $S \cap U = \phi$, then $|S|>k$, which implies that $\norm{f(S)}_{\infty}>1$ using condition $(D)$. If $S \cap U \neq \phi$, then $S \notin \mathcal{S}^{\downarrow}$ which implies that $\norm{f(S)}_{\infty}>1$ using condition $(B)$. Thus, $\norm{f(S)}_{\infty}\leq 1$ if and only if $S \in \mathcal{T}^{\downarrow}$. 
\end{proof}

We are now ready to prove our main embedding result i.e.~\Cref{thm:bounded-dimension}.

\begin{proof}[Proof of Theorem~\ref{thm:bounded-dimension}]
We define a graph $G=(V, E)$ as follows: two elements $u,v \in V$ are adjacent in $G$ if there exist sets $S_1, S_2 \in \mathcal{S}$ (not necessarily distinct) such that $u \in S_1, v \in S_2, S_1 \cap S_2 \neq \emptyset$.
As the cardinality of each set in $\mathcal{S}$ is at most $k$ and each element of $V$ is present in at most $\Delta$ sets, the maximum degree of a vertex in $G$ can be bounded above as 
\[
\Delta(G) \leq k(k-1)\Delta^2
\]
Thus, the chromatic number of $G$ is at most $L=\chi(G)\leq k(k-1)\Delta^2+1 \leq k^2\Delta^2$. Using the greedy coloring algorithm, we can partition $V$ into $L$ non-empty parts $U_1, U_2, \ldots, U_L$ such that each $U_j$ is a independent set in $G$.
For every $j \in [L]$, as $U_j$ is an independent set in $G$, we have
\begin{enumerate}
    \item For every set $S \in \mathcal{S}$, $|S \cap U_j|\leq 1$. 
    \item Any two sets $S_1, S_2 \in \mathcal{S}$ with $S_1 \cap U_j \neq \emptyset$, $S_2 \cap U_j \neq \emptyset$ and $S_1 \cap S_2 \neq \emptyset$ satisfy $S_1 \cap U_j = S_2 \cap U_j = S_1 \cap S_2$.
\end{enumerate}

We now define the set families $\mathcal{S}_1, \mathcal{S}_2, \ldots, \mathcal{S}_L$ as follows: 
\[
\mathcal{S}_j = \{ S \in \mathcal{S} : S \cap U_j \neq \emptyset \} \cup \{ S \subseteq V \setminus U_j : |S| \leq k \}
\]
We claim that $\bigcap_{j \in [L]}\mathcal{S}_j^{\downarrow} = \mathcal{S}^{\downarrow}$. 
First, consider an arbitrary set $S \in \mathcal{S}^{\downarrow}$ and an integer $j \in [L]$.
As $|S|\leq k$, irrespective of $S$ intersects $U_j$ or not, $S \in \mathcal{S}_j^{\downarrow}$.
Thus, $\mathcal{S}^{\downarrow} \subseteq \mathcal{S}_j^\downarrow$ for all $j \in [L]$. Consider a non-empty set $S \notin \mathcal{S}^{\downarrow}$. As $U_1, U_2, \ldots, U_L$ is a partition of $V$, there exists a $j \in [L]$ such that $S \cap U_j \neq \emptyset$. As $S \notin \mathcal{S}^{\downarrow}$, $S \notin \mathcal{S}_j^{\downarrow}$. This implies that 
\[
\bigcap_{j \in [L]}\mathcal{S}_j^{\downarrow} = \mathcal{S}^{\downarrow}
\]
Using~\Cref{prop:dim-intersection}, in order to bound the packing dimension of $\mathcal{S}^{\downarrow}$, it suffices to bound the packing dimension of $\mathcal{S}_j^{\downarrow}$, $j \in [L]$. 

Fix an integer $j \in [L]$ and consider the set family $\mathcal{S}_j^{\downarrow}$. It is defined on the universe $V$ and there exists a non-empty subset $U_j \subseteq V$ such that 
\[
\mathcal{S}_j = \mathcal{S}'_j \cup \{ S \subseteq V \setminus U_j : |S| \leq k \}
\]
 with 
 \[
 \mathcal{S}'_j = \{ S \in \mathcal{S} : S \cap U_j \neq \emptyset \}.
 \]
Here, $\mathcal{S}'_j$ is a simple set system which satisfies the following properties:
\begin{enumerate}
    \item Each set in $\mathcal{S}'_j$ has cardinality at most $k\geq 2$ and each element appears in at most $\Delta$ sets in $\mathcal{S}'_j$.
    \item Every set $S \in \mathcal{S}'_j$ satisfies $|S \cap U_j|=1$. As $\mathcal{S}$ is non-trivial, for every $u \in U_j$, there exists a set $S \in \mathcal{S}'_j$ with $u \in S$. 
    \item For every pair of sets $S_1, S_2 \in \mathcal{S}'_j$ with $S_1 \cap S_2 \neq \phi$, $S_1 \cap U_j = S_2 \cap U_j=S_1 \cap S_2$.
 \end{enumerate}

In other words, the set family $\mathcal{S}_j'$ is a sunflower-bouquet with core $U_j$. 
Using~\Cref{cor:dimension}, we get that $\textsf{pdim}(\mathcal{S}_j^{\downarrow}) \leq (k\Delta)^{O(1)}$ for all $j \in [L]$, which completes the proof. 
\end{proof}

\subsection{Hardness of Vector Bin Packing}
\label{subsec:vbp-hardness}
We show that for large enough constant $d$, Vector Bin Packing is hard to approximate within $\Omega(\log d)$. Our hardness is obtained via the hardness of set cover on simple bounded instances. 

In the set cover problem, the input is a set family $\mathcal{S}$ on a universe $V$ with $|V|=n$. The objective is to pick the minimum number of sets $\{ S_1, S_2, \ldots, S_m\} \subseteq \mathcal{S}$ from the family such that their union is equal to $V$. The greedy algorithm where we repeatedly pick the set that covers the maximum number of new elements achieves a $\ln n$ approximation factor. Fiege~\cite{Feige98} proved a matching hardness of $(1-\epsilon)(\ln n)$. On set systems where each pair of sets intersect in at most one element i.e. simple instances, $\Omega(\log n)$ hardness of set cover is proved by Kumar, Arya, and Ramesh~\cite{KAR00}. We observe that by changing the parameters slightly, their reduction also implies the same hardness on instances where the maximum set size is bounded:

\begin{theorem}(Set Cover on simple bounded instances) 
\label{thm:bounded-set-cover}
There exists an integer $B_0$ such that for every constant $B \geq B_0$, the Set Cover problem on simple set systems in which each set has cardinality at most $B$ is NP-hard to approximate within $\Omega(\log B)$. Furthermore, in the hard instances, each element occurs in at most $O(B)$ sets. 
\end{theorem}
The details of the parameter modification appear in~\Cref{sec:simple}.
\vspace{0.02\textwidth}

We combine this set cover hardness with the bound on the packing dimension of simple set systems to prove the hardness of Vector Bin Packing. 
%\begin{theorem}
%\label{thm:vbp-main}
%There exists an integer $d_0$ such that for every constant $d \geq d_0$, Vector Bin Packing on $d$-dimensional vectors is NP-hard to approximate within $\Omega(\log d)$. 
%\end{theorem}
\begin{proof}[Proof of Theorem~\ref{thm:main-vbp}]
We prove the result by giving an approximation preserving reduction from the NP-hard problem of set cover on simple bounded set systems. Let $\mathcal{S}$ be the set system from~\Cref{thm:bounded-set-cover} defined on a universe $V$. Note that each set in the family has cardinality at most $k=B$ and each element in the universe appears in at most $\Delta = O(B)$ sets. 
We now output a set $\mathcal{V}$ of $|V|$ vectors in $[0,1]^d$ such that 
\begin{enumerate}
    \item (Completeness.) If there is a set cover of size $m$ in $\mathcal{S}$, there is a packing of $\mathcal{V}$ using $m$ bins. 
    \item (Soundness.) If there is no set cover of size $m'$ in $\mathcal{S}$, there is no packing of $\mathcal{V}$ using $m'$ bins.  
\end{enumerate}

We use~\Cref{thm:bounded-dimension} to compute an embedding $f: V \rightarrow [0,1]^d$ in polynomial time such that
\[
\norm{f(S)}_{\infty} \leq 1
\]
if and only if $S \in \mathcal{S}^{\downarrow}$, with $d =(k\Delta)^{O(1)}=B^{O(1)}$. 
Our output Vector Bin Packing instance is the set of vectors $f(v), v \in V$. 
\[
\mathcal{V} = \{ f(v) : v \in V\}
\]

\paragraph{Completeness.} Suppose that there exist sets $S_1, S_2, \ldots, S_m\in \mathcal{S}$ whose union is $V$. Then, we use $m$ bins with the vectors $\{f(v_j) :  j \in S_i\}$ in the $i$th bin. A vector might appear in multiple bins, but we can arbitrarily pick one bin for each vector while still maintaining the property that in each bin, the $\ell_{\infty}$ norm of the sum of the vectors is at most $1$. 

\paragraph{Soundness.} Suppose that the minimum set cover in $\mathcal{S}$ has cardinality at least $m'+1$. Then, we claim that the set of vectors $\mathcal{V}$ needs $m'+1$ bins to be packed. Suppose for contradiction that there is a vector packing with $m'$ bins. In other words, there exists a partition of $V$ into $B_1, B_2, \ldots, B_{m'}$ such that $\norm{f(B_i)}_{\infty} \leq 1$ for all $i \in [m']$. As $\norm{f(B_i)}_{\infty}\leq 1$, $B_i \in \mathcal{S}^{\downarrow}$ for all $i \in [m']$. That is, for every $i \in [m']$, there exists a set $S_i \in \mathcal{S}$ such that $B_i \subseteq S_i$. This implies that $\{ S_1, S_2, \ldots, S_{m'} \}$ is a set cover of $V$, a contradiction. 

As the original bounded simple set cover problem is hard to approximate within $\Omega(\log B)=\Omega(\log d)$, the resulting Vector Bin Packing is hard to approximate within $\Omega(\log d)$. Furthermore, in the hard instances, the optimal value i.e. the minimum number of bins needed to pack the vectors can be made arbitrarily large, and thus, the hardness applies to the asymptotic approximation ratio.
\end{proof}
\section{Vector Scheduling}
\label{sec:vs}
\subsection{Monochromatic Clique}
In the Monochromatic Clique problem, given a graph $G=([n],E)$ and a parameter $k(n)$, the objective is to assign $k$ colors to the vertices of $G$ so as to minimize the largest monochromatic clique.  
More formally, we study the following decision version of the problem. 
\begin{definition}(\emph{\mc}$(k,B)$)
	In the \emph{\mc}$(k,B)$  problem, given a graph $G=(V,E)$ with $|V|=n$ and parameters $k(n),B(n)$, the goal is to distinguish between the following: 
	\begin{enumerate}
		\item (YES case) The chromatic number of $G$ is at most $k$.
		\item (NO case) In any assignment of $k$ colors to the vertices of $G$, there is a clique of size $B$, all of whose vertices are assigned the same color.
	\end{enumerate}
\end{definition}
It generalizes the standard $k$-Coloring problem, which corresponds to the case when $B=2$. Note that the problem gets easier as $B$ increases. 
Indeed, when $B > \sqrt{n}$, we can solve the problem in polynomial time using the canonical SDP relaxation. 
We present this algorithm and an almost matching integrality gap in~\Cref{sec:appendix-vs}. 

On the hardness front, we now prove that \mc$(k,B)$ is hard when $B=(\log n)^{C}$, for any constant $C$. We achieve this in two steps: First, we observe that the existing chromatic number hardness results already imply the hardness of monochromatic clique when $B=(\log n)^{\gamma}$ for some constant $\gamma>0$. Next, we amplify this hardness by using lexicographic graph product.  
%Finally, we prove~\Cref{thm:main-vs} using the reduction from Monochromatic Clique to Vector Scheduling as in~\cite{CK04}.

\subsubsection{Basic Hardness}
We start with a couple of  basic Ramsey theoretic lemmas from~\cite{CK04}. 
\begin{lemma}
	\label{lem:ramsey-independent}
	For a graph $G=(V,E)$ with $|V|=n$, if $\omega(G)\leq B$, then $\alpha(G)\geq n^{\frac{1}{B}}$.
\end{lemma}
%We can repeatedly peel off the independent sets to argue about the chromatic number.
\begin{lemma}
	\label{lem:ramsey}
	For a graph $G=(V,E)$ with $|V|=n$, if $\omega(G)\leq B$, then $\chi(G)\leq O(n^{1-\frac 1B} \log n)$.
\end{lemma}

We can use the above lemmas to prove that if the chromatic number of a graph is large enough, then in any assignment of $k$ colors to the vertices of the graph, there is a large monochromatic clique. 
\begin{lemma}
	\label{lem:color-clique}
	For every constant $\epsilon>0$, if a graph $G=(V,E)$ with $|V|=n$ satisfies $\chi(G)\geq k\frac{n}{2^{\left( \log n\right)^{\alpha}}}$ for some integer $k$ and $0<\alpha <1$, then in any assignment of $k$ colors to $V$, there is a monochromatic clique of size $B=\Omega\left((\log n)^{1-\alpha-\epsilon}\right)$.
\end{lemma}
\begin{proof}
	Suppose for contradiction that there is an assignment of $k$ colors $V$ without a monochromatic clique of size $B$. Using~\Cref{lem:ramsey}, the subgraphs corresponding to each of the $k$ color classes has chromatic number at most 
	\[
	O(n^{1-\frac 1B} \log n) = \frac{n}{2^{\Omega\left((\log n)^{\alpha+\epsilon}\right)}}\log n<\frac{n}{2^{(\log n)^{\alpha}}}
	\]
	colors. Thus, the whole graph has chromatic number at most $k\frac{n}{2^{(\log n)^{\alpha}}}$ colors, a contradiction.
\end{proof}

Khot~\cite{Khot01} proved that assuming $\textsf{NP} \nsubseteq \textsf{ZPTIME}\left( n^{(\log n)^{O(1)}}\right)$, the chromatic number of graphs is hard to approximate within a factor of $\frac{n}{2^{\left(\log n\right)^{1-\gamma}}}$ for an absolute constant $\gamma>0$. More formally, he proved the following:
\begin{theorem}(~\cite{Khot01})
	\label{thm:chromatic-hardness}
	There exists a constant $\gamma > 0$, a function $k=k(n)$, and a randomized reduction that takes as input a $3$-SAT instance $I$ on $n$ variables and outputs a graph $G=(V,E)$ with $|V|=N=2^{\log n ^{O(1)}}$ such that  
	\begin{enumerate}
		\item (Completeness) If $I$ is satisfiable, $\chi(G)\leq k$.
		\item (Soundness) If $I$ is not satisfiable, with probability at least $\frac 12$, $\chi(G)> k\frac{N}{2^{\left(\log N\right)^{1-\gamma}}}$.
	\end{enumerate}
	Futhermore, the reduction runs in time $\emph{\textsf{poly}}(N)=2^{ (\log n)^{O(1)}}$.
\end{theorem}
We observe that Khot's chromatic number hardness immediately gives $(\log n)^{\Omega(1)}$ hardness of Monochromatic Clique.
\begin{lemma}
	\label{lem:monochromatic-clique-hard-basic}
	There exists a constant $\gamma >0$, a function $k=k(n)$ such that the following holds. Assuming $\textsf{NP} \nsubseteq \textsf{ZPTIME}\left(n^{ (\log n)^{O(1)}}\right)$, given a graph $G=([n],E)$, there is no $n^{(\log n)^{O(1)}}$ time algorithm for \emph{\mc}$(k,B)$ when $B=\Omega\left((\log n)^{\gamma}\right)$.
\end{lemma}
\begin{proof}
	Using Khot's reduction, we get that there exists an absolute constant $\gamma>0$ such that assuming $\textsf{NP} \nsubseteq \textsf{ZPTIME}\left(n^{ (\log n)^{O(1)}}\right)$, given a graph $G=([n],E)$ and a parameter $k(n)$, there is no $n^{(\log n)^{O(1)}}$ time algorithm to distinguish between the following: 
	\begin{enumerate}
		\item (Completeness) $\chi(G)\leq k$.
		\item (Soundness) $\chi(G) > k \frac{n}{2^{\left(\log n\right)^{1-\gamma}}}$.
	\end{enumerate}
	Using~\Cref{lem:color-clique}, the Soundness condition implies that in any assignment of $k$ colors to $G$, there is a monochromatic clique of size $\Omega\left( (\log n)^{\gamma -\epsilon}\right)$, for any constant $\epsilon>0$. Thus, given a graph $G$ and a parameter $k$, assuming $\textsf{NP} \nsubseteq \textsf{ZPTIME}\left(n^{ (\log n)^{O(1)}}\right)$, there is no $n^{(\log n)^{O(1)}}$ time algorithm to distinguish between the following: 
	\begin{enumerate}
		\item (Completeness) $\chi(G)\leq k$.
		\item (Soundness) In any assignment of $k$ colors to the vertices of $G$, there is a monochromatic clique of size $\Omega( (\log n)^{\gamma'})$.
	\end{enumerate}
	for any constant $\gamma' < \gamma$.
\end{proof}

\subsubsection{Amplification using Lexicographic Product}
We cannot directly amplify the hardness of the Monochromatic-Clique problem by taking graph products as we cannot preserve the chromatic number and also amplify the largest clique in an assignment of $k$ colors at the same time.
We get around this issue by defining a harder variant of Monochromatic Clique called Strong Monochromatic Clique and then amplifying it. 

\begin{definition}(\emph{\smc}$(k,B,C)$)
	In the \emph{\smc}$(k,B,C)$, given a graph $G$ and parameters $k(n),B(n),C$, the goal is to distinguish between the following two cases:
	\begin{enumerate}
		\item (YES case) The chromatic number of $G$ is at most $k$.
		\item (NO case) In any assignment of $k^C$ colors to the vertices of $G$, there is a monochromatic clique of size $B$.
	\end{enumerate}
\end{definition}

We now observe that the chromatic number hardness of Khot~\cite{Khot01} implies the same hardness as~\Cref{lem:monochromatic-clique-hard-basic} for Strong Monochromatic Clique as well. 
\begin{lemma}
	\label{lem:smc-hardness}
	There exists a constant $\gamma >0$ and a function $k=k(n)$ such that for every constant $C\geq 1$, the following holds. Assuming $\emph{\textsf{NP}} \nsubseteq \emph{\textsf{ZPTIME}}\left(n^{ (\log n)^{O(1)}}\right)$, there is no $n^{(\log n)^{O(1)}}$ time algorithm for \emph{\smc}$(k,B,C)$ when $B=\Omega( (\log n)^\gamma)$.
\end{lemma}
\begin{proof}
	Note that the function $k$ in~\Cref{thm:chromatic-hardness} satisfies $k=o\left(2^{(\log N)^{1-\gamma}}\right)$. Thus, we can replace the soundness condition in~\Cref{thm:chromatic-hardness} with $\chi(G)\geq k^C\frac{N}{2^{C\left(\log N\right)^{1-\gamma}}}$. Using ~\Cref{lem:color-clique}, this implies that in any assignment of $k^C$ colors to the vertices of $G$, there is a monochromatic clique of size $\Omega( (\log N)^{\gamma -\epsilon})$, where $\epsilon >0$ is an absolute constant. The hardness of Strong Monochromatic Clique then follows along the same lines as~\Cref{lem:monochromatic-clique-hard-basic}.
\end{proof}

We amplify the hardness of \smc$(k,B,C)$ to \mc$(k^C,B^C)$ using the lexicographic product of graphs. 
First, we define lexicographic product and prove some properties of it. 
\begin{definition}(Lexicographic product of graphs)
	Given two graphs $G$ and $H$, the Lexicographic graph product $G \cdot H$ has vertex set $V(G) \times V(H)$, and two vertices $(u_1,v_1), (u_2,v_2)$ are adjacent if either $(u_1, u_2)\in E(G)$ or $u_1=u_2$ and $(v_1, v_2)\in E(H)$. 
\end{definition}
The lexicographic product can be visualized as replacing each vertex of $G$ with a copy of $H$ and forming complete bipartite graphs between copies of vertices adjacent in $G$.
For ease of notation, we let $G^2 = G \cdot G$. More generally, for an integer $n$ that is a power of $2$, we define $G^n$ as taking the above lexicographic product of $G$ with itself recursively $\log n$ times. 
\begin{lemma}
	\label{lem:chromatic-lexi}
	Let $n\geq 2$ be a power of $2$. If $\chi(G)\leq k$, then $\chi(G^n)\leq k^n$. 
\end{lemma}
\begin{proof}
	We prove that $\chi(G^2)\leq k^2$, and the statement follows by induction on $n$.
	If $f:G\rightarrow [k]$ is a proper $k$-coloring of $G$, then the coloring $f'(u,v)=(f(u),f(v))$ is a proper $k^2$-coloring of $G \times G$.
\end{proof}
\begin{lemma}
	\label{lem:monochromatic-lexi}
	Let $n \geq 2$ be a power of $2$. Suppose that in any assignment of $k$ colors to the vertices of $G$, there is a monochromatic clique of size $B$. Then, in any assignment of $k$ colors to the vertices of $G^n$, there is a monochromatic clique of size $B^n$.
\end{lemma}
\begin{proof}
	We prove the statement for $n=2$ and the lemma follows by induction on $n$. Let $f:V(G^2)\rightarrow [k]$ be a given assignment. For a vertex $v \in G$, consider the assignment $g_v:V(G)\rightarrow[k]$ defined as $g_v(u)=f(v,u)$. As every assignment of $k$ colors to the vertices of $G$ has a monochromatic clique of size $B$, there is a color $\alpha(v) \in [k]$ and a clique $S(v) \subseteq V(G)$ with $|S(v)|\geq B$ such that $g_v(u)=\alpha(v)$ for all $u \in S(v)$, or in other words, $f(v,u)=\alpha(v)$ for all $u \in S(v)$. Note that such a set $S(v)$ and $\alpha(v)$ exist for $v \in V(G)$. The function $\alpha : V(G)\rightarrow [k]$ can also be visualized as an assignment of $k$ colors to the vertices of $G$, and thus there is a monochromatic clique $T$ of size at least $B$ with respect to this assignment. The set 
	\[
	\{ S(v) : v \in T \}
	\]
	is a monochromatic clique of size $B^2$ with respect to $f$ in $G$. 
\end{proof}

By using the lexicographic product, we can get a polynomial time reduction from Strong Monochromatic Clique to Monochromatic Clique. 
\begin{lemma}
	\label{lem:smc-reduction}
	For every constant $C \geq 1$ that is a power of $2$, there exists a polynomial time reduction from \emph{\smc}$(k,B,C)$ to \emph{\mc}$(k^C,B^C)$.
\end{lemma}
\begin{proof}
	Given a graph $G$ as an instance of \smc$(k,B,C)$, we compute the graph $G'=G^C$. We claim that solving \mc $(k^C,B^C)$ on $G'$ solves the original Strong Monochromatic Clique problem. 
	\begin{enumerate}
		\item (Completeness.) Suppose that $\chi(G)\leq k$. Then, by~\Cref{lem:chromatic-lexi}, $\chi(G')\leq k^C$.
		\item (Soundness.) Suppose that in any assignment of $k^C$ colors to the vertices of $G$, there is a monochromatic clique of size $B$. Then, by~\Cref{lem:monochromatic-lexi}, in any assignment of $k^C$ colors to the vertices of $G'$, there is a monochromatic clique of size $B^C$.\qedhere
	\end{enumerate}
\end{proof}

Putting everything together, we obtain the following hardness of Monochromatic Clique.
\begin{theorem}
	\label{thm:monochromatic-clique-hard}
	For every constant $C> 0$, there exists a function $k=k(n)$ such that the following holds. Assuming $\textsf{NP} \nsubseteq \textsf{ZPTIME} \left( n^{(\log n)^{O(1)}} \right) $, there is no $n^{ (\log n)^{O(1)}}$ time algorithm for \emph{\mc}$(k,B)$ when $B=\Omega\left((\log n)^C\right)$.
\end{theorem}
\begin{proof}
	The proof follows directly by combining~\Cref{lem:smc-hardness} and~\Cref{lem:smc-reduction}.
\end{proof}

\subsection{From Monochromatic Clique to Vector Scheduling}
We now prove~\Cref{thm:main-vs} using the above hardness of Monochromatic Clique.
\begin{proof}[Proof of Theorem ~\ref{thm:main-vs}]
	The reduction from \mc$(k,B)$ to Vector Scheduling is (implicitly) proved in~\cite{CK04}. We present it here for the sake of completeness. Given a graph $G=(V=[n],E)$, parameters $k$ and $B$, we order all the $B$-sized cliques of $G$ as $T_1, T_2, \ldots, T_d$ with $d\leq n^B$.
	We define a set of $n$ vectors $v_1, v_2, \ldots, v_n$ of dimension $d$ with 
	\[
	(v_i)_j = \begin{cases}
		1& \text{if }i \in T_j \\ 
		0& \text{otherwise.}
	\end{cases}
	\]
	The instance of the Vector Scheduling has these $n$ vectors as the input and the number of machines is equal to $k$. 
	
	We analyze the reduction. 
	\begin{enumerate}
		\item (Completeness.) Suppose that there exists a proper $k$-coloring of $G$, $c:V\rightarrow [k]$. We assign the vector $v_i$ to the machine $c(i)$. For every $j \in [d]$, all the $B$ vectors that have $1$ in the $j$th dimension are assigned to distinct machines.
		Thus, the makespan of the scheduling is at most $1$. 
		\item (Soundness.) Suppose that in any assignment of $k$ colors to the vertices of $G$, there is a monochromatic clique of size $B$. In this case, the makespan of the scheduling is at least $B$.
	\end{enumerate}
	We set $B=(\log n)^C$ for a large constant $C$ to be set later. 
	We choose $k$ from~\Cref{thm:monochromatic-clique-hard} such that assuming $\textsf{NP} \nsubseteq \textsf{ZPTIME} \left( n^{(\log n)^{O(1)}} \right) $, there is no $n^{(\log n)^{O(1)}}$ time algorithm for \mc$(k,B)$. By the above reduction, we can conclude that there is no polynomial time algorithm that approximates the resulting Vector Scheduling instances within a factor of $B=(\log n)^C$. 
	As $d \leq n^B$, we get that $\log d \leq (\log n)^{C+1}$, and $B \geq (\log d)^{1-\frac{1}{C+1}}$. Setting $C=\frac{1}{\epsilon}-1$, we get that $d$-dimensional Vector Scheduling has no polynomial time $\Omega( (\log d)^{1-\epsilon})$ approximation algorithm assuming $\textsf{NP} \nsubseteq \textsf{ZPTIME}\left(n^{(\log n)^{O(1)}}\right)$, for every constant $\epsilon >0$. 
\end{proof}
\begin{remark}
	In~\cite{IKKP19}, Im, Kell, Kulkarni, and Panigrahi also study the $\ell_r$-norm minimization of Vector Scheduling where the objective is to minimize 
	\[
	\max_{k \in [d]} \left( \sum_{i=1}^m (L_i^k)^r \right)^{\frac{1}{r}}
	\]
	where $L_i^k$ denotes the load on the machine $i$ on the $k$th dimension. 
	They gave an algorithm with an approximation ratio $O\left(\left( \frac{\log d}{\log\log d}\right)^{1-\frac{1}{r}}\right)$. Our reduction from Monochromatic Clique gives almost optimal hardness for this variant as well: we get the hardness of $\Omega\left(\left( \log d\right)^{1-\frac{1}{r}-\epsilon}\right)$ assuming $\textsf{NP}\nsubseteq \textsf{ZPTIME}\left(n^{(\log n)^{O(1)}}\right)$, for every constant $\epsilon>0$.
\end{remark}

\subsection{Hardness of Vector Scheduling via Balanced Hypergraph Coloring}
Observe that the resulting Vector Scheduling instances in the above reduction satisfy a stronger property: the vectors are from $\{ 0,1\}^d$. In the setting where the vectors are from $\{ 0,1\}^d$, the Vector Scheduling problem is closely related to the Balanced Hypergraph Coloring problem. In this problem, given a hypergraph $H$ and an integer $k$, the objective is to assign $k$ colors to the vertices of $H$ minimizing the maximum number of monochromatic vertices in an edge. 
More formally, we study the following decision version of the problem. 
\begin{definition} (Balanced Hypergraph Coloring.)
	In the Balanced Hypergraph Coloring problem, given a $s$-uniform hypergraph $H$ and parameters $k$ and $c<s$, the objective is to distinguish between the following: 
	\begin{enumerate}
		\item There is an assignment of $k$ colors to the vertices of $H$ such that in every edge, each color appears at most $c$ times. 
		\item The hypergraph $H$ has no proper coloring with $k$ colors i.e., in any assignment of $k$ colors to the vertices of $H$, there is an edge all of whose $s$ vertices are assigned the same color. 
	\end{enumerate}
\end{definition}
We give a simple reduction from Balanced Hypergraph Coloring to Vector Scheduling. 
\begin{lemma}
	\label{lem:coloring-vs}
	Given a $s$-uniform hypergraph $H=(V'=[n'],E')$ and parameters $k, c$, there is a polynomial time reduction that outputs a Vector Scheduling instance $I$ over $n'$ vectors $v_1, v_2,\ldots, v_{n'} \in \{0,1\}^d$ on $m'$ machines with $m'=k, d=|E'|$ such that 
	\begin{enumerate}
		\item (Completeness.) If there is an assignment of $k$ colors to the vertices of $H$ such that each color appears at most $c$ times in every edge, then there is a scheduling of $I$ with makespan at most $c$. 
		\item (Soundness.) If $H$ has no proper coloring with $k$ colors, then in any scheduling of $I$, the makespan is at least $s$. 
	\end{enumerate}
\end{lemma}
\begin{proof}
	Let $d=|E'|$.
	Order the edges of the hypergraph $H$ as $e_1, e_2, \ldots, e_d$. 
	We define the set of vectors $v_1, v_2, \ldots, v_{n'} \in \{0,1\}^d$ as follows: 
	\[
	(v_i)_j = \begin{cases}
		1& \text{if }i \in e_j \\ 
		0& \text{ otherwise.}
	\end{cases}
	\]
	We set the number of machines $m'$ to be equal to the number of colors $k$. There is a natural correspondence between the assignment of $k$-colors to the vertices of $H$ $f:V'\rightarrow [k]$, and the scheduling where we assign the vector $v_i$ to the machine $f(i)$. We now analyze our reduction. 
	\begin{enumerate}
		\item (Completeness.) If there exists an assignment of $k$ colors $f:V' \rightarrow [k]$ where each color appears at most $c$ times in each edge, we assign the vector $v_i, i \in [n']$ to the machine $f(i)$. In any dimension $j \in [d]$, at most $c$ vectors $v_i$ with $(v_i)_j = 1$ are scheduled on any machine. Thus, in any machine, the total load in each dimension is at most $c$. 
		\item (Soundness.) If there exists a vector scheduling $f:[n]\rightarrow [m']$ with makespan strictly smaller than $s$, assign the color $f(i)$ to the $i$th vertex of the hypergraph. In any edge of the hypergraph, each color appears fewer than $s$ times as the makespan is smaller than $s$. Thus, $f:V' \rightarrow [k]$ is a proper $k$-coloring of the hypergraph $H$. \qedhere
	\end{enumerate}
\end{proof}

We prove the hardness results for Vector Scheduling, namely~\Cref{thm:main-vs-np} and~\Cref{thm:main-vs-intermediate} by combining this reduction with the hardness of Balanced Hypergraph Coloring.
Note that the dimension of the resulting instances in the above reduction is equal to $m$, the number of edges in the hypergraph $H$, and the ratio of the makespans in the completeness and soundness is equal to $\frac{s}{c}$. Thus, our goal is to prove the hardness of the Balanced Hypergraph Coloring problem where $\frac{s}{c}$ is as large as possible, as a function of $m$, the number of edges in the underlying hypergraph. 

Towards this, we first give a reduction from the Label Cover problem to the Balanced Hypergraph Coloring problem.
\begin{lemma}
	\label{lem:lc-hypergraph}
	Fix an odd prime number $k \geq 3$ and let $\epsilon = \frac{1}{k^8}$. Given a Label Cover instance $G=(V=L\cup R,E,\Sigma_L, \Sigma_R,\Pi)$, there is a polynomial time reduction that outputs a $k^2$ uniform hypergraph $H=(V',E')$ with $|V'|\leq |L|k^{|\Sigma_L|}$ such that 
	\begin{enumerate}
		\item(Completeness) If $G$ is satisfiable, there is an assignment of $k$ colors to the vertices of $H$ such that in every edge, each color occurs at most $2k$ times. 
		\item(Soundness) If no labeling to $G$ can satisfy an $\epsilon$ fraction of the constraints, then $H$ has no proper $k$-coloring, that is, in any assignment of $k$ colors to the vertices of $H$, there is an edge all of whose vertices are assigned the same color.
	\end{enumerate}
	Furthermore, $|E'|$ is at most $|R|\Delta^kk^{|\Sigma_L|k^2}$ where $\Delta$ is the maximum degree of a vertex $v \in R$.
\end{lemma}
We defer the proof of~\Cref{lem:lc-hypergraph} to~\Cref{subsection:lc-hypergraph}. 

\medskip 
Using~\Cref{lem:lc-hypergraph}, we can prove the hardness of Balanced Hypergraph Coloring via Label Cover hardness results.
We obtain two different hardness results for the Balanced Hypergraph Coloring problem, one under $\textsf{NP}\nsubseteq \textsf{DTIME}\left( n^{O(\log \log n)}\right)$ and another NP-hardness result, by using two different hardness results for the Label Cover problem.
These two hardness results prove~\Cref{thm:main-vs-intermediate} and~\Cref{thm:main-vs-np} respectively, using~\Cref{lem:coloring-vs}.

First, using the standard Label Cover hardness obtained using PCP Theorem~\cite{ALMSS98} combined with Raz's Parallel Repetition theorem~\cite{Raz98}, we get the following hardness of Balanced Hypergraph Coloring. 
\begin{theorem}
	\label{thm:bhc-parallel}
	Assuming $\textsf{NP} \nsubseteq \textsf{DTIME}\left( n^{O(\log \log n)}\right)$, there is no polynomial time algorithm for the following problem. 
	Given a $k^2$-uniform hypergraph $H=(V',E')$ with $m =|E'|$ and $k=(\log m)^{\Omega(1)}$, distinguish between the following: 
	\begin{enumerate}
		\item There is an assignment of $k$ colors to the vertices of $H$ such that in any edge of the hypergraph, each color appears at most $2k$ times. 
		\item The hypergraph $H$ has no proper $k$ coloring. 
	\end{enumerate}
\end{theorem}

\begin{proof}
	By setting $\epsilon = \frac{1}{k^8}$ in~\Cref{thm:lc-parallel}, we have a reduction from the $3$-SAT problem on $n$ variables to the Label Cover problem $G=(V= L \cup R, E, \Sigma_L,\Sigma_R,\Pi)$ with soundness $\epsilon$ and $|V|\leq n^{O(\log k)}$, $|\Sigma_L|\leq k^{O(1)}$ and $\Delta \leq k^{O(1)}$.
	Using~\Cref{lem:lc-hypergraph}, we can reduce this Label Cover instance to a Balanced Hypergraph Coloring instance $H=(V',E')$ with $|V'|\leq n^{O(\log k)}2^{k^{O(1)}}$ and $|E'|\leq n^{O(\log k)}2^{k^{O(1)}}$.
	We set $k = (\log n)^{\Omega(1)}$ such that $|V'| = n ^ {O(\log \log n)}$ and $|E'| = n^{O(\log \log n)}$ to obtain the required hardness of Balanced Hypergraph Coloring. 
\end{proof}

The proof of~\Cref{thm:main-vs-intermediate} follows immediately from~\Cref{thm:bhc-parallel} and~\Cref{lem:coloring-vs}.

\medskip 

Next, using the hardness of near linear sized Label Cover due to Moshkovitz and Raz~\cite{MR10}, we obtain the following NP-hardness of Balanced Hypergraph Coloring. 
\begin{theorem}
	\label{thm:bhc-nearlinear}
	For any constant $C \geq 1$, given a $k^2$ uniform hypergraph $H=(V',E')$ with $m=|E'|$ and $k=(\log \log m)^C$, it is NP-hard to distinguish between the following:
	\begin{enumerate}
		\item There is an assignment of $k$ colors to the vertices of $H$ such that in any edge of the hypergraph, each color appears at most $2k$ times. 
		\item The hypergraph $H$ has no proper $k$ coloring. 
	\end{enumerate}
\end{theorem}
\begin{proof}
	By setting $\epsilon = \frac{1}{k^8}$ in~\Cref{thm:linear-lc}, we can reduce a $3$-SAT instance over $n$ variables to a Label Cover instance $G=(V = L \cup R, E, \Sigma_L, \Sigma_R, \Pi)$ with soundness $\epsilon$ and $|V|\leq n^{1+o(1)}k^{O(1)}, |\Sigma_L|\leq 2^{k^{O(1)}}$, $\Delta=k^{O(1)}$.
	By using~\Cref{lem:lc-hypergraph}, we can reduce the Label Cover instance to a Balanced Hypergraph Coloring instance $H=(V',E')$ with $|V'|\leq n^{1+o(1)}2^{2^{k^{O(1)}}}$ and $|E'|$ at most $ n^{1+o(1)}2^{2^{k^{O(1)}}} $. We set $k = (\log \log n)^{\Omega(1)}$ to obtain $|V'|=O(n^2), |E'|=O(n^2)$. 
	
	Dinur and Steurer~\cite{DinurS14} gave an improvement to~\cite{MR10}--in the new Label Cover hardness, the alphabet size $|\Sigma_L|$ can be taken to be $2^{\left( \frac{1}{\epsilon}\right)^\gamma}$ for \textit{every} constant $\gamma>0$. Using this improved Label Cover hardness, we can set $k = (\log \log n)^{C}$ for any constant $C\geq 1$ in the hardness of Balanced Hypergraph Coloring. 
\end{proof}

The proof of~\Cref{thm:main-vs-np} follows immediately from~\Cref{thm:bhc-nearlinear} and~\Cref{lem:coloring-vs}. 

\medskip 

Finally, we remark that if the structured graph version of the Projection Games Conjecture~\cite{Moshkovitz15} holds,~\Cref{lem:lc-hypergraph} and~\Cref{lem:coloring-vs} together prove that $d$-dimensional Vector Scheduling is NP-hard to approximate within a factor of $(\log d)^{\Omega(1)}$. 

\subsection{Proof of~\Cref{lem:lc-hypergraph}}
\label{subsection:lc-hypergraph}
We follow the standard Label Cover-Long Code framework--see e.g.,\cite{ABP20}.

\smallskip \noindent \textbf{Reduction.}
For ease of notation, let $n = |\Sigma_L|$. 
For every node $v\in L$ of the Label Cover instance, we have a set of $k^{n}$ vertices denoted by $f_v=\{v\}\times [k]^{n}$. The vertex set of the hypergraph is $V'=\bigcup_{v \in L}f_v$.

For every $u \in R$, and $k$ distinct neighbors of $u$, $v_1, v_2, \ldots, v_k \in L$ with projection constraints $\pi_i:[\Sigma_L]\rightarrow[\Sigma_R], i \in [k]$, consider the set of $k^2$ vectors $\textbf{x}^{i,j}$ for $i\in [k], j \in [k]$ which satisfy the following: For every $\beta \in \Sigma_R$, and for all $\alpha_1, \alpha_2, \ldots, \alpha_k \in \Sigma_L$ such that $\pi_i(\alpha_i)=\beta$ for all $i \in [k]$, we have 
	\begin{equation}
		\label{eq:lc-constraint}
		\left| \{ (i,j) | \textbf{x}^{i,j}_{\alpha_i} =p \}\right| \leq 2k \, \, \forall p \in [k] 
	\end{equation}
For every such set of $k^2$ vectors, we add the edge $\{(v_i, \textbf{x}^{i,j}):1 \leq i,j \leq k\}$ to $E'$.
We can observe that $|V'|\leq |L|k^{|\Sigma_L|}$ and 
\[
|E'|\leq |R|\binom{\Delta}{k}\binom{k^{|\Sigma_L|}}{k}^k\leq |R|\Delta^kk^{|\Sigma_L|k^2}.
\]

\smallskip \noindent \textbf{Completeness.}
Suppose that there exists an assignment $\sigma : V \rightarrow \Sigma$ that satisfies all the constraints of the Label Cover instance $G$. We color the set of vertices $f_v$ in the long code corresponding to the vertex $v \in L$ with the dictator function on the coordinate $\sigma(v)$ i.e. for every $\textbf{x} \in f_v$, we assign the color 
\[
c\left(\{v,\textbf{x}\}\right)=\textbf{x}_{\sigma(v)}
\]

We can observe that this coloring satisfies the property that in every edge $e \in E'$, each color appears at most $2k$ times. 

\smallskip \noindent \textbf{Soundness.}
Suppose that there is a proper $k$-coloring $c:V'\rightarrow [k]$ of the hypergraph $H$ i.e. in every edge $e=\{v_1, v_2, \ldots, v_{k^2}\}$, we have 
\[
\left| \{ c(v_1), c(v_2), \ldots, c(v_{k^2})\} \right| >1
\]
Our goal is to prove that there is a labeling to the Label Cover instance that satisfies at least $\epsilon = \frac{1}{k^8}$ fraction of constraints.

We need the following lemma proved by Austrin, Bhangale, Potukuchi~\cite{ABP20} using a generalization of Borsuk-Ulam theorem.
\begin{lemma}(Theorem $5.2$ of~\cite{ABP20})
	\label{lem:exceptional-cooridinates}
	For every odd prime $k$ and $n \geq k^3$, in any $k$-coloring of $[k]^n$, $c:[k]^n \rightarrow [k]$, there is a set of $k$ vectors $\textbf{x}^1, \textbf{x}^2, \ldots, \textbf{x}^k$ that are all assigned the same color such that 
	\[
	\{ \textbf{x}^1_i,\textbf{x}^2_i, \ldots, \textbf{x}^k_i  \} = [k]
	\]
	for at least $n-k^3$ distinct coordinates $i \in [n]$.
\end{lemma}

Using this lemma, for every $v \in L$, we can identify a set of vectors $\textbf{x}^{v,1},\textbf{x}^{v,2},\ldots,\textbf{x}^{v,k} \in f_v$ such that all these vectors have the same color  i.e. $c(\{v,\textbf{x}^{v,i}\})=c'(v)$ for all $v \in L, i \in [k]$ for some function $c':L\rightarrow [k]$.
Furthermore, there are a set of coordinates $S(v)\subseteq [n]$ with $|S(v)|\leq k^3$ such that 
\[
\{ \textbf{x}^{v,1}_i , \textbf{x}^{v,2}_i, \ldots, \textbf{x}^{v,k}_i \} = [k]
\]
for every $i \in [n]\setminus S(v)$.

For a set $S \subseteq \Sigma_L$ and a function $\pi : \Sigma_L \rightarrow \Sigma_R$, we use $\pi(S)$ to denote the set $\{ \pi(i) : i \in S \}$. 
We now prove a key lemma that helps in the decoding procedure.
\begin{lemma}
	\label{lem:disjoint-union}
	Let $u \in R$ be a node on the right side of the Label Cover instance. There are a set of labels $S(u)\subseteq \Sigma_R$ such that $|S(u)|\leq k^5$, and for every $v\in L$ that is a neighbor of $u$ with projection constraint $\pi :\Sigma_L \rightarrow \Sigma_R$, we have $S(u)\cap \pi(S(v))\neq \phi$. 
\end{lemma}

\begin{proof}
	Fix a node $u \in R$ on the right side of the Label Cover instance. 
	Let $v_1, v_2, \ldots, v_l \in L$ be the neighbors of $u$ in the Label Cover instance corresponding to the projection constraints $\pi_1, \pi_2, \ldots, \pi_l$ respectively. 
	As $|S(v_i)|\leq k^3$ for all $i \in [l]$, and the constraints $\pi_i$ are projections, we have $|\pi_i(S(v_i))|\leq k^3$ for all $i \in [l]$. Among these $l$ subsets $\pi_i(S(v_i))$ of $\Sigma_R$, let the maximum number of pairwise disjoint subsets be denoted by $l'$. Without loss of generality, we can assume that $\mathcal{S}=\{\pi_i(S(v_i)) : i \in [l']\}$ is a pairwise disjoint family of subsets. 
	
	We define the set $S(u)$ as follows: 
	\[
	S(u)= \bigcup_{i \in [l']}\pi_i(S(v_i))
	\]
	As $\mathcal{S}$ is a family of maximum pairwise disjoint subsets, we have $S(u)\cap \pi_i(S(v_i))\neq \phi$ for all $i \in [l]$. Our goal is to bound the size of $S(u)$, which we achieve by bounding $l'$. 
	
	We claim that $l' \leq  k(k-1)$. Suppose for contradiction that $l' > k(k-1)$. This implies that there are $l' > k(k-1)$ nodes $v_1, v_2, \ldots, v_{l'}$ all adjacent to $u$ such that $\pi_i(S(v_i)), i \in [l']$ are all pairwise disjoint. Thus, there exists a color $\ell \in [k]$ and a set of $k$ nodes $w_1, w_2, \ldots, w_k$ adjacent to $u$ corresponding to the projection constraints $\pi'_1, \pi'_2, \ldots, \pi'_k$ such that $c'(w_i)=\ell$ for all $i \in [k]$, and the $k$ sets $\pi'_i(S(w_i))$ are pairwise disjoint. 
	
	Using this, we can construct a set of vectors $\textbf{x}^{i,j}, 1\leq i,j \leq k$ defined as $\textbf{x}^{i,j}=\textbf{x}^{w_i,j}$ which satisfy the following properties: 
	\begin{enumerate}
		\item All these vectors are colored the same: 
		\[
		c(\{w_i,\textbf{x}^{i,j}\})=\ell \,\, \forall 1 \leq i,j \leq k
		\]
		\item For every $i \in [k]$, 
		\[
		\{ \textbf{x}^{i,1}_{i'} , \textbf{x}^{i,2}_{i'}, \ldots, \textbf{x}^{i,k}_{i'} \} = [k]
		\]
		for every $i' \in [n]\setminus S(w_i)$.
	\end{enumerate}
	We claim that these set of vectors satisfy the condition in~\Cref{eq:lc-constraint}. Fix a $\beta \in \Sigma_R$, and $\alpha_1, \alpha_2, \ldots, \alpha_k \in \Sigma_L$ such that $\pi'_i(\alpha_i)=\beta$ for all $i \in [k]$. As the family of subsets $\pi'_i(S(w_i))$ is a pairwise disjoint family, we can infer that there exists at most one $i\in [k]$ such that $\alpha_i \in S(w_i)$. Note that if $\alpha_i \notin S(w_i)$, then 
	\[
	\{ \textbf{x}_{\alpha_i}^{i,j} : j \in [k] \} = [k].
	\]
	Thus, we have 
	\[
	\left| \{ (i,j) | \textbf{x}^{i,j}_{\alpha_i} =p \}\right| \leq 2k \, \, \forall p \in [k].
	\]
	Thus, the set of vectors $\{ (w_i, \textbf{x}^{i,j}): 1\leq i,j\leq k \}$ is indeed an edge of $E'$. As all these vectors are colored the same color $\ell$, we have arrived at a contradiction to the fact that $c$ is a proper $k$-coloring of $H$. 
	
	Hence, we can conclude that $l'\leq k(k-1)$, and thus, $|S(u)|\leq k(k-1)k^3<k^5$. 
\end{proof} 

Now, consider the labeling $\sigma : L \rightarrow \Sigma_L,$ where $\sigma(v), v \in L$ is chosen uniformly at random from $S(v)$. Similarly, let $\sigma : R \rightarrow \Sigma_R$ is chosen uniformly at random from $S(u), u \in R$. Using~\Cref{lem:disjoint-union}, we can infer that for every edge $e=(v,u)$ in the Label Cover, this labeling satisfies the edge $e$ with probability at least $\frac{1}{|S(v)||S(u)|}\geq \frac{1}{k^8}$. By linearity of expectation, this labeling satisfies at least $\frac{1}{k^8}$ fraction of the constraints in expectation. Hence, with positive probability, the labeling satisfies at least $\frac{1}{k^8}$ fraction of the constraints. This concludes the proof of soundness that if $H$ has a proper $k$ coloring, then there exists a labeling to $G$ that satisfies at least $\frac{1}{k^8}$ fraction of the constraints.

\section{Vector Bin Covering}
\label{sec:vbc}
%We prove that vector bin covering is hard to approximate within $\Omega\left(\frac{\log d}{\log \log d}\right)$. Our hardness result is obtained via hardness of rainbow coloring of hypergraphs. 

\subsection{Hardness of Vector Bin Covering via Rainbow Coloring}

As is the case with the Vector Scheduling problem, the hard instances for Vector Bin Covering are when the vectors are from $\{0,1\}^d$. In this setting, the Vector Bin Covering problem is closely related to the hypergraph rainbow coloring problem. 
A hypergraph $H=(V,E)$ is said to be $k$-rainbow colorable if there is an assignment of $k$ colors to the vertices of $H$ such that in every edge, all the $k$ colors appear. 
When $k=2$, it is equivalent to the standard $2$-coloring of hypergraphs. 
Unlike the usual (hyper)graph coloring, rainbow coloring gets harder with larger number of colors. 

In the approximate rainbow coloring problem, given a hypergraph that is promised to have a rainbow coloring with a large number of colors, the goal is to find a coloring in polynomial time using fewer number of colors. 
More formally, the computational problem we study is the following.
\begin{definition}(Approximate Rainbow Coloring)
	In the approximate rainbow coloring problem, the input is a hypergraph $H=(V,E)$ and a parameter $k\geq 2$. The objective is to distinguish between the following: 
	\begin{enumerate}
		\item The hypergraph $H$ is $k$-rainbow colorable. 
		\item The hypergraph $H$ has no valid $2$-coloring. 
	\end{enumerate}
\end{definition}

We now give a simple reduction from approximate rainbow coloring to Vector Bin Covering. 
\begin{lemma}
	\label{lem:rainbow-vbc}
	Given a hypergraph $H=(V,E)$ and a parameter $k$, there is a polynomial time reduction that outputs a Vector Bin Covering instance $v_1, v_2, \ldots, v_n \in \{ 0, 1\}^d$ with $n = |V|, d=|E|$ such that 
	\begin{enumerate}
		\item (Completeness.) If $H$ is $k$-rainbow colorable, there is a partition of $[n]$ into $k$ parts $A_1, A_2, \ldots, A_k$ such that 
		\[
			\sum_{j \in A_i}v_j \geq \textbf{1}^d \,\,\forall i \in [k]	
		\]
		\item (Soundness.) If $H$ is not $2$-colorable, there is no partition of $[n]$ into $A_1, A_2$ such that 
		\[
			\sum_{j \in A_i}v_j \geq \textbf{1}^d \,\,\forall i \in [2]	
		\]
	\end{enumerate}
\end{lemma}
\begin{proof}
	Let $n=|V|, d = |E|$. We order the edges $E$ as $E = \{ e_1, e_2, \ldots, e_d\}$. 
	We output a set of vectors $\mathcal{V} = \{ v_1, v_2, \ldots, v_n\}$ where each $v_i \in \{ 0,1\}^d$ is defined as follows: 
	\[
	(v_i)_j = \begin{cases}
		1& \text{if }i \in e_j \\ 
		0& \text{ otherwise.}
	\end{cases}
	\]
	We analyze this reduction: 
	\begin{enumerate}
		\item (Completeness.) Suppose that the hypergraph $H$ has a rainbow coloring with $k$ colors $f : V \rightarrow [k]$. We partition $[n]$ into $k$ parts $A_1, A_2, \ldots, A_k$ such that 
		\[
		A_i = \{ j \in [n] : f(j)=i\}
		\]
		Consider an arbitrary integer $i \in [k]$. Note that for every edge $e$ in $H$, $e \cap A_i \neq \phi$. 
		Thus, 
		\[
		\sum_{j \in A_i} v_j \geq \textbf{1}^d
		\]
		\item (Soundness.) Suppose that the hypergraph $H$ has no proper coloring with $2$ colors. Then, we claim that there is no partition of $[n]$ into two parts $A_1, A_2$ such that 
		\[
		\sum_{j \in A_i} v_j \geq \textbf{1}^d \,\, \forall i \in [2]
		\]
		Suppose for contradiction that there exists $A_1, A_2$ with the above property. Consider the coloring of the hypergraph $f: V \rightarrow [2]$ as 
		\[
		f(v) = \begin{cases}
			1& \text{if }v \in A_1 \\ 
			2& \text{if }v \in A_2
		\end{cases}
		\]
		Consider an arbitrary edge $e_l, l \in [d]$ of the hypergraph $H$. As $\sum_{j \in A_i}(v_j)_l \geq 1$ for all $i \in [2]$, there exist $v_1, v_2 \in e_l$ such that $v_1 \in A_1, v_2 \in A_2$. Thus, the coloring $f$ is a proper $2$ coloring of the hypergraph $H$, a contradiction. \qedhere
		
	\end{enumerate}
	
\end{proof}
We combine this reduction with the hardness of approximate rainbow coloring to prove the hardness of Vector Bin Covering, namely~\Cref{thm:main-vbc}. 
Note that the dimension of the resulting vectors in the Vector Bin Covering instance $\mathcal{V}$ is equal to the number of edges $m=|E|$ of the hypergraph $H$, and the gap in the optimal Bin Covering value of $\mathcal{V}$ is equal to $k$, the number of colors. 
Hence, to obtain better inapproximability results for Vector Bin Covering that grow with $d$, our goal is to show the hardness of approximate rainbow coloring on hypergraphs with $m$ edges where the number of colors $k$ is as large a function of $m$ as possible. Towards this, we prove that it is NP-hard to $2$-color a hypergraph with $m$ edges that is promised to be rainbow colorable with $k=\Omega\left(\frac{\log m}{\log \log m}\right)$ colors. 
\begin{theorem}
	\label{thm:rainbow-coloring}
	Given a hypergraph $H$ with $m$ edges, it is NP-hard to distinguish between the following:
	\begin{enumerate}
		\item (Completeness) $H$ is $k$-rainbow colorable. 
		\item (Soundness) $H$ is not $2$-colorable. 
	\end{enumerate}
	where $k=\Omega\left( \frac{\log m}{\log \log m}\right)$. 
\end{theorem}

We defer the proof of~\Cref{thm:rainbow-coloring} to~\Cref{subsection:rainbow-coloring}.

\medskip 

We now prove the hardness of Vector Bin Covering using~\Cref{thm:rainbow-coloring}. 
\begin{proof}[Proof of Theorem~\ref{thm:main-vbc}]
Using~\Cref{thm:rainbow-coloring} combined with the reduction in~\Cref{lem:rainbow-vbc}, we get that the following problem is NP-hard. Given a set of $n$ vectors $v_1, v_2, \ldots, v_n \in \{0,1\}^d$, distinguish between
\begin{enumerate}
	\item $\mathcal{V}$ can be partitioned into $k=\Omega\left( \frac{\log d}{\log \log d}\right)$ parts such that in each part, the sum of vectors is at least $1$ in every coordinate. 
	\item $\mathcal{V}$ cannot be partitioned into $2$ parts such that in each part, the sum of vectors is at least $1$ in every coordinate. In other words, the maximum number of parts into which $\mathcal{V}$ can be partitioned such that in each part, the sum of vectors is at least $1$ in every coordinate is equal to $1$. 
\end{enumerate}
Thus, it is NP-hard to approximate $d$-dimensional Vector Bin Covering within $k = \Omega\left( \frac{\log d}{\log \log d}\right)$. 
\end{proof}

\subsection{Proof of~\Cref{thm:rainbow-coloring}}
\label{subsection:rainbow-coloring}
Our proof follows by viewing the hypergraph rainbow coloring problem as a promise constraint satisfaction problem(PCSP) and analyzing its polymorphisms~\cite{AGH17,BG16,BKO19}. The idea is to prove that the polymorphisms have a small number of ``important'' coordinates which can then be decoded in the Label Cover-Long Code framework. For the case of the above rainbow coloring PCSP, we prove that the polymorphisms are $1$-fixing in that there is a single coordinate which when set to a certain value fixes the value of the function. This characterization then implies the hardness of the approximate rainbow coloring.

For ease of readability, we skip defining polymorphisms formally, and instead present the proof as a simple gadget reduction from Label Cover.  
We first need a definition. 
\begin{definition}($1$-fixing~\cite{BG16,GS20})
	A function $f:[k]^n \rightarrow \{0,1\}$ is said to be $1$-fixing if there exists an index $\ell \in [n]$ and values $\alpha, \beta \in [k]$ such that 
	\[
	f(\textbf{x})= 0 \, \forall \textbf{x} : \textbf{x}_\ell = \alpha \quad\text{and}\quad f(\textbf{x})= 1 \,\forall \textbf{x} : \textbf{x}_\ell = \beta 
	\]
\end{definition}
In the analysis of the gadget, we need a definition and a lemma from~\cite{ABP20}.
\begin{definition}(The hypergraph $H_r^n[k]$)
	The hypergraph $H_r^n[k] =(V,E)$ is a $k$-uniform hypergraph with vertex set as the set of $n$-dimensional vectors over $[k] $ i.e. $V=[k]^n$. A set of $k$ vectors $\textbf{v}^1, \textbf{v}^2, \ldots, \textbf{v}^k$ form an edge of the hypergraph if 
	\[
	\sum_{i=1}^n \left|[k]\setminus \{ \textbf{v}^j_i : j \in [k] \} \right| \leq r
	\]
\end{definition}
\begin{lemma}
	\label{lem:abp20}
	For every $k \geq 2$, the hypergraph $H_{\lfloor \frac k2 \rfloor}^n[k]$ is not $2$-colorable.
\end{lemma}
We analyze the gadget used in our reduction.
\begin{lemma}
	\label{lem:gadget}
	Fix $k \geq 3$. Suppose $f:[k]^n \rightarrow \{ 0,1\}$ satisfies the below two-coloring property: For every $2k$ vectors $\textbf{v}^1, \textbf{v}^2, \ldots, \textbf{v}^{2k} \in [k]^n$ with 
	\[
	\{ \textbf{v}^j_i : j \in [2k] \}  = [k] \, \forall i \in [n],
	\]
	we have
	\[
	\{ f(\textbf{v}^j) : j \in [2k] \} = \{0,1\}.
	\]
	Then, $f$ is $1$-fixing. 
\end{lemma}
\begin{proof}
	We first prove that there exist $\ell \in [n], \alpha \in [k]$, $b \in \{0,1\}$ such that $f(\textbf{x})=b$ for all $\textbf{x} \in [k]^n$ with $\textbf{x}_\ell = \alpha$.
	Suppose for contradiction that this is not the case. Then, for every $i \in [n], j \in [k]$ there exist vectors $\textbf{x}^{i,j}, \textbf{y}^{i,j} \in [k]^n$ such that $\textbf{x}^{i,j}_i = \textbf{y}^{i,j}_i =j$, and $f(\textbf{x}^{i,j})=0$ where as $f(\textbf{y}^{i,j})=1$.
	
	Let $r=\lfloor \frac k2 \rfloor$. We view $f:[k]^n \rightarrow \{0,1\}$ as an assignment of two colors to the vertices of the hypergraph $H_r^n [k]$. As the hypergraph is not two colorable (\Cref{lem:abp20}), we can infer that there is an edge of $H_r^n[k]$ all of whose vertices are assigned the same color.
	In other words, there exist $k$ vectors $\textbf{v}^1, \textbf{v}^2, \ldots, \textbf{v}^k \in [k]^n$ and $b \in \{0,1\}$ such that $f(\textbf{v}^j)=b$ for all $j \in [k]$. Furthermore, there are at most $r$ missing values in these vectors i.e. 
	\[
	\sum_{i=1}^n \left|[k]\setminus \{ \textbf{v}^j_i : j \in [k] \} \right| \leq r
	\]
	Now, we pick $r$ vectors $\textbf{u}^1, \textbf{u}^2,\ldots, \textbf{u}^r$ (with repetitions if needed) by filling the missing values using $\textbf{x}^{i,j},\textbf{y}^{i,j}$ vectors such that 
	\begin{enumerate}
		\item $f(\textbf{u}^j)=b$ for all $j \in [r]$. 
		\item For every $i \in [n]$, 
		\[
		\{ \textbf{v}^j_i : j \in [k] \} \cup \{ \textbf{u}^j_i : j \in [r] \} = [k]
		\]
	\end{enumerate}
	By taking the union of $\{ \textbf{v}^1, \textbf{v}^2, \ldots, \textbf{v}^k\}$ and $\{\textbf{u}^1, \textbf{u}^2, \ldots, \textbf{u}^r\}$, and repeating some vectors, we obtain $2k$ vectors  $  \textbf{w}^1, \textbf{w}^2, \ldots, \textbf{w}^{2k}$ with $f(\textbf{w}^j)=b$ for all $j \in [2k]$, and 
	\[
	\{ \textbf{w}^j_i : j \in [2k] \} = [k] \, \, \forall i \in [n]
	\]
	However, this contradicts the two-coloring property of $f$. Thus, there exist $\ell \in [n], \alpha \in [k], b \in \{0,1\}$ such that $f(\textbf{x})=b$ for all $\textbf{x} \in [k]^n$ with $\textbf{x}_\ell =\alpha$. 
	
	We now claim that there exists $\beta \in [k]$ such that $f(\textbf{x})=1-b$ for all $\textbf{x} \in [k]^n$ with $\textbf{x}_\ell =\beta$. Suppose for contradiction that this is not the case. Then, there exist $k$ vectors $\textbf{v}^1, \textbf{v}^2, \ldots, \textbf{v}^k$ such that $\textbf{v}^j_\ell = j$ for all $j \in [k]$, and $f(\textbf{v}^j)=b$ for all $j \in [k]$. We now pick $\textbf{v}^{k+1}, \textbf{v}^{k+2}, \ldots, \textbf{v}^{2k} \in [k]^n$ such that $\textbf{v}^j_\ell = \alpha$ for all $j \in \{k+1, k+2, \ldots, 2k\}$, and $\textbf{v}^j_i = j -k$ for all $i \in [n]$ with $i \neq \ell$, and $j \in \{k+1, k+2, \ldots, 2k\}$. These $2k$ vectors $\textbf{v}^1, \textbf{v}^2, \ldots, \textbf{v}^{2k}$ satisfy 
	\begin{enumerate}
		\item $f(\textbf{v}^j)=b$ for all $j \in [2k]$. 
		\item For every $i \in [n]$, 
		\[
		\{\textbf{v}^j_i : j \in [2k] \} = [k] 
		\]
	\end{enumerate}
	contradicting the two-coloring property of $f$. Thus, there exists $\beta \in [k]$ such that $f(\textbf{x})=1-b$ for all $\textbf{x} \in [k]^n$ with $\textbf{x}_\ell = \beta$, completing the proof that $f$ is $1$-fixing. 
\end{proof}

We are now ready to prove~\Cref{thm:rainbow-coloring}. Our hardness result is obtained using a reduction from the Label Cover problem. This reduction is standard in the PCSP literature.(See e.g., ~\cite{BKO19})

\smallskip \noindent \textbf{Reduction.}
We start with the Label Cover instance $G=(V=L\cup R),E,\Sigma=\Sigma_L=\Sigma_R,\Pi)$ from~\Cref{thm:lc-hardness} and output a hypergraph $H=(V',E')$.
Let $n$ denote the label size $n = |\Sigma|$. 
For each vertex $v \in L \cup R$, we have a long code containing a set of nodes $K_v$ of size $[k]^n$, indexed by $n$ length vectors. 
\begin{enumerate}
	\item The vertex set of the hypergraph $V'$ is the union of all the long code nodes. 
	\[
	V' =\bigcup_{v \in V}K_v
	\]
	\item Edges of the hypergraph: For every vertex $v \in V$ of the Label Cover instance, we add an edge in $E'$ for each set of $2k$ vectors $\{\textbf{v}^1, \textbf{v}^2, \ldots, \textbf{v}^{2k}\}$ in $K_v$, if 
	\begin{equation}
		\label{eq:constraints}
		\{ \textbf{v}^j_i : j \in [2k] \}  = [k] \, \forall i \in [n].
	\end{equation}
	The number of edges in $H$ is at most 
	\[
	|E'|\leq |V|\binom{k^{|\Sigma|}}{2k}\leq |V|k^{O(k)}
	\]
	\item Equality constraints: For every constraint $\Pi_e : u \rightarrow v$ of the Label Cover, we add a set of equality constraints between nodes $\textbf{x} \in K_u$, $\textbf{y} \in K_v$ if for all $i \in [n]$, $\textbf{x}_i = \textbf{y}_{\Pi_e(i)}$. By adding an equality constraint between two nodes, we identify the two nodes together and treat it as a single node. That is, we compute the connected components of the equality constraints graph and identify a single master node for each component. We then obtain a multi-hypergraph $H_1$ from $H$ by replacing each node with the corresponding master node. However, a vertex could appear multiple times in an edge in $H_1$. We delete such occurrences from $H_1$ by setting each edge to be a simple set of the vertices contained in it, and obtain the final hypergraph $H_2$. We note the following: 
	\begin{enumerate}
		\item There exists a $k$-rainbow coloring of $H$, $f:V' \rightarrow [k]$ that respects the equality constraints i.e. $f(\textbf{x})=f(\textbf{y})$ for all pairs of nodes $\textbf{x},\textbf{y}$ with equality constraints between them if and only if $H_2$ is $k$-rainbow colorable. 
		\item Similarly, there exists a $2$-coloring of $H$ that respects equality constraints if and only if $H_2$ is $2$-colorable.
	\end{enumerate}
	Finally, the number of edges in $H_2$ is at most the number of edges in $H$.
\end{enumerate}

\smallskip \noindent \textbf{Completeness.} Suppose that there is a labeling $\sigma : V \rightarrow \Sigma$ that satisfies all the constraints. We define the coloring $f: V' \rightarrow [k]$ of $H$ as follows. For every node $\textbf{x} \in K_v$, we set the dictatorship function
\[
f(\textbf{x})=\textbf{x}_{\sigma(v)}
\]
By the constraints added in~\Cref{eq:constraints}, the function $f$ is a valid $k$-rainbow coloring of $H$.
As $\sigma$ satisfies all the constraints of the Label Cover, the coloring $f$ satisfies all the equality constraints. 

\smallskip \noindent \textbf{Soundness.} Suppose that there is no labeling $\sigma : V \rightarrow \Sigma$ that satisfies all the constraints in $G$. Then we claim that there is no $2$-coloring of $H$ that respects all the equality constraints. Suppose for contradiction that there is a $2$-coloring $f:V' \rightarrow \{0,1\}$ that respects all the equality constraints. 

Consider a vertex $v \in V$. The function $f_v:[k]^n\rightarrow \{0,1\}$, defined as $f$ on $K_v$ satisfies the conditions in~\Cref{lem:gadget}. Thus, $f_v$ is $1$-fixing for every $v \in V$. Hence, there is a function $L: V \rightarrow \Sigma$ such that for every $v \in V$, $f_v$ is $1$-fixing on the coordinate $L(v)$. We now claim that the labeling $\sigma : V \rightarrow \Sigma$ defined as $\sigma(v) =L(v)$ satisfies all the constraints in $G$. 

Consider an edge $e=(u,v)$, $u\in L, v \in R$ with the projection constraint $\Pi_e : \Sigma \rightarrow \Sigma$. Our goal is to show that $\Pi_e(L(u))=L(v)$. Suppose for contradiction that $\Pi_e(L(u)) \neq L(v)$. 
By the $1$-fixing property of $f_u$, we have $\alpha_u, \beta_u \in [k]$ such that  
\[
f_u(\textbf{x})= 0 \, \forall \textbf{x} \in [k]^n: \textbf{x}_{L(u)} = \alpha_u \quad\text{and}\quad f_u(\textbf{x})= 1 \,\forall \textbf{x} \in [k]^n: \textbf{x}_{L(u)} = \beta_u 
\]
Similarly, we have $\alpha_v, \beta_v \in [k]$ such that 
\[
f_v(\textbf{y})= 0 \, \forall \textbf{y}\in [k]^n : \textbf{y}_{L(v)} = \alpha_v \quad\text{and}\quad f_v(\textbf{y})= 1 \,\forall \textbf{y} \in [k]^n: \textbf{y}_{L(v)} = \beta_v 
\]
By the equality constraints, $f_u(\textbf{x})=f_v(\textbf{y})$ for all $\textbf{x}, \textbf{y}\in [k]^n$ such that $\textbf{x}_i = \textbf{y}_{\Pi_e(i)}\, \forall i \in [n]$. Let $\textbf{y}' \in [k]^n$ be an arbitrary vector with $ \textbf{y}'_{\Pi_e(L(u))} = \alpha_u, \textbf{y}'_{L(v)}=\beta_v$. We choose $\textbf{x}' \in [k]^n$ such that for all $i \in [n]$, $\textbf{x}'_i = \textbf{y}'_{\Pi_e(i)}$. Note that $\textbf{x}'_{L(u)}=\alpha_u$. Thus, $f_u(\textbf{x}')=0$ where as $f_v(\textbf{y}')=1$. However, this contradicts the equality constraints.

\section{Conclusion}
\label{sec:conclusion}

We conclude by mentioning a few open problems.

\begin{enumerate}
	\item A drawback of our hardness result for Vector Bin Packing is that our lower bound is only applicable when $d$ is large enough.  The reason our hardness result needs $d$ to be large enough is that the $k$-set cover hardness itself~\cite{Trevisan01} needs $k$ to be large enough. By starting our reduction with alternate set cover hardness results such as the $3$-dimensional matching problem, we can obtain improved hardness for smaller values of $d$. However, this approach will still not help for $d=2$ as the upper bound on the packing dimension of these set families is greater than $2$. It is an interesting open problem to characterize set families with packing dimension $2$. We believe that such a characterization could help in proving the hardness of approximation results for the set cover problem on set families with packing dimension $2$, which directly gives improved hardness of $2$-dimensional Vector Bin Packing where there is a significant gap between the $1.406$ factor algorithm~\cite{BEK06} and the hardness factor of $1.00167$~\cite{ray2021aptas}.
	\item $d$-dimensional Geometric Bin Packing is another open problem where packing dimension based ideas could help. The best hardness result for the $d$-dimensional Geometric Bin Packing is still the (tiny) hardness factor from the $2$-dimensional setting. A possible avenue to obtain improved inapproximability for this problem is by a reduction from a suitable set cover variant using a notion of packing dimension. However, this is easier said than done as the geometric packings are significantly harder to tame--for example, it is NP-hard to decide if a given set of $n$ rectangles can fit in a unit square, while the corresponding problem for Vector Bin Packing is trivial.
	\item Another interesting open problem is to close the gap between $O(\log d)$ algorithm and $\Omega \left( \frac{\log d}{\log \log d}\right)$ hardness for Vector Bin Covering. For the rainbow coloring, the question is: Given a hypergraph $H$ with $m$ edges that is promised to be $f(m)$-rainbow colorable, can we $2$-color it in polynomial time? The answer to this question when $f(m)$ is equal to $O(\log m)$ is yes, by a simple random $2$-coloring. By~\Cref{thm:rainbow-coloring}, the answer is no when $f(m)=\Omega\left( \frac{\log m}{\log \log m}\right)$. Which of these is tight? Bhattiprolu, Guruswami, and Lee~\cite{BhattiproluGL15} proved that in certain settings, the simple random coloring is optimal for rainbow coloring. This suggests that perhaps even here, the problem is hard when $f(m)=o(\log m)$. On the other hand, obtaining any hardness result beyond $\Omega\left( \frac{\log m}{\log \log m}\right)$ seems to be impossible with the Label Cover-Long Code framework. 
\end{enumerate}

\section*{Acknowledgements}

I am greatly indebted to Venkatesan Guruswami for helpful discussions, for his detailed feedback on the manuscript which significantly improved the presentation, and for his encouragement. 
I also thank Nikhil Bansal, Varun Gupta, Ravishankar Krishnaswamy, and Janani Sundaresan for helpful discussions and feedback on the manuscript. I am especially grateful to Ravishankar Krishnaswamy for pointing me to~\cite{KAR00}.
\bibliography{ref}

\newcommand{\etalchar}[1]{$^{#1}$}
\begin{thebibliography}{BOVvdZ16}

\bibitem[AAC{\etalchar{+}}98]{AlonACESVW98}
Noga Alon, Yossi Azar, J{\'{a}}nos Csirik, Leah Epstein, Sergey~V. Sevastianov,
  Arjen P.~A. Vestjens, and Gerhard~J. Woeginger.
\newblock On-line and off-line approximation algorithms for vector covering
  problems.
\newblock {\em Algorithmica}, 21(1):104--118, 1998.

\bibitem[ABP19]{ABP19}
Per Austrin, Amey Bhangale, and Aditya Potukuchi.
\newblock Simplified inpproximability of hypergraph coloring via t-agreeing
  families.
\newblock {\em CoRR}, abs/1904.01163, 2019.

\bibitem[ABP20]{ABP20}
Per Austrin, Amey Bhangale, and Aditya Potukuchi.
\newblock Improved inapproximability of rainbow coloring.
\newblock In {\em Proceedings of the 2020 {ACM-SIAM} Symposium on Discrete
  Algorithms, {SODA} 2020}, pages 1479--1495, 2020.

\bibitem[ACKS13]{AzarCKS13}
Yossi Azar, Ilan~Reuven Cohen, Seny Kamara, and F.~Bruce Shepherd.
\newblock Tight bounds for online vector bin packing.
\newblock In {\em Symposium on Theory of Computing Conference, STOC'13}, pages
  961--970, 2013.

\bibitem[AGH17]{AGH17}
Per Austrin, Venkatesan Guruswami, and Johan H{\aa}stad.
\newblock (2+{\(\epsilon\)})-sat is {NP}-hard.
\newblock {\em {SIAM} J. Comput.}, 46(5):1554--1573, 2017.

\bibitem[AJKL84]{Assman84}
Susan~F. Assmann, David~S. Johnson, Daniel~J. Kleitman, and Joseph~Y.{-}T.
  Leung.
\newblock On a dual version of the one-dimensional bin packing problem.
\newblock {\em J. Algorithms}, 5(4):502--525, 1984.

\bibitem[ALM{\etalchar{+}}98]{ALMSS98}
Sanjeev Arora, Carsten Lund, Rajeev Motwani, Madhu Sudan, and Mario Szegedy.
\newblock Proof verification and the hardness of approximation problems.
\newblock {\em J. {ACM}}, 45(3):501--555, 1998.

\bibitem[Ban17]{Bansal17}
Nikhil Bansal.
\newblock Scheduling open problems: Old and new.
\newblock {\em MAPSP 2017}, www.mapsp2017.ma.tum.de/MAPSP2017-Bansal.pdf, 2017.

\bibitem[BCS09]{BansalCS09}
Nikhil Bansal, Alberto Caprara, and Maxim Sviridenko.
\newblock A new approximation method for set covering problems, with
  applications to multidimensional bin packing.
\newblock {\em {SIAM} J. Comput.}, 39(4):1256--1278, 2009.

\bibitem[BEK16]{BEK06}
Nikhil Bansal, Marek Eli{\'{a}}{\v{s}}, and Arindam Khan.
\newblock Improved approximation for vector bin packing.
\newblock In {\em Proceedings of the Twenty-Seventh Annual {ACM-SIAM} Symposium
  on Discrete Algorithms, {SODA} 2016}, pages 1561--1579, 2016.

\bibitem[BG16]{BG16}
Joshua Brakensiek and Venkatesan Guruswami.
\newblock New hardness results for graph and hypergraph colorings.
\newblock In {\em 31st Conference on Computational Complexity, {CCC} 2016},
  pages 14:1--14:27, 2016.

\bibitem[BGL15]{BhattiproluGL15}
Vijay V. S.~P. Bhattiprolu, Venkatesan Guruswami, and Euiwoong Lee.
\newblock Approximate hypergraph coloring under low-discrepancy and related
  promises.
\newblock In {\em Approximation, Randomization, and Combinatorial Optimization.
  Algorithms and Techniques, {APPROX/RANDOM} 2015}, pages 152--174, 2015.

\bibitem[Bha18]{Bhangale18}
Amey Bhangale.
\newblock {NP}-{H}ardness of coloring 2-colorable hypergraph with
  poly-logarithmically many colors.
\newblock In {\em 45th International Colloquium on Automata, Languages, and
  Programming, {ICALP} 2018}, pages 15:1--15:11, 2018.

\bibitem[BK14]{BansalK14}
Nikhil Bansal and Arindam Khan.
\newblock Improved approximation algorithm for two-dimensional bin packing.
\newblock In {\em Proceedings of the Twenty-Fifth Annual {ACM-SIAM} Symposium
  on Discrete Algorithms, {SODA} 2014}, pages 13--25, 2014.

\bibitem[BKO19]{BKO19}
Jakub Bul{\'{\i}}n, Andrei~A. Krokhin, and Jakub Oprs\v{a}l.
\newblock Algebraic approach to promise constraint satisfaction.
\newblock In {\em Proceedings of the 51st Annual {ACM} {SIGACT} Symposium on
  Theory of Computing, {STOC} 2019}, pages 602--613, 2019.

\bibitem[BOVvdZ16]{BansalOVZ16}
Nikhil Bansal, Tim Oosterwijk, Tjark Vredeveld, and Ruben van~der Zwaan.
\newblock Approximating vector scheduling: Almost matching upper and lower
  bounds.
\newblock {\em Algorithmica}, 76(4):1077--1096, 2016.

\bibitem[BS04]{BansalS04}
Nikhil Bansal and Maxim Sviridenko.
\newblock New approximability and inapproximability results for 2-dimensional
  bin packing.
\newblock In {\em Proceedings of the Fifteenth Annual {ACM-SIAM} Symposium on
  Discrete Algorithms, {SODA} 2004}, pages 196--203, 2004.

\bibitem[CC06]{ChlebikC06}
Miroslav Chleb{\'{\i}}k and Janka Chleb{\'{\i}}kov{\'{a}}.
\newblock Inapproximability results for orthogonal rectangle packing problems
  with rotations.
\newblock In {\em Algorithms and Complexity, 6th Italian Conference, {CIAC}
  2006}, pages 199--210, 2006.

\bibitem[CFS08]{ChandranFS08}
L.~Sunil Chandran, Mathew~C. Francis, and Naveen Sivadasan.
\newblock Boxicity and maximum degree.
\newblock {\em J. Comb. Theory, Ser. {B}}, 98(2):443--445, 2008.

\bibitem[CJK01]{CsirikJK01}
J{\'{a}}nos Csirik, David~S. Johnson, and Claire Kenyon.
\newblock Better approximation algorithms for bin covering.
\newblock In {\em Proceedings of the Twelfth Annual Symposium on Discrete
  Algorithms, 2001}, pages 557--566, 2001.

\bibitem[CK04]{CK04}
Chandra Chekuri and Sanjeev Khanna.
\newblock On multidimensional packing problems.
\newblock {\em {SIAM} J. Comput.}, 33(4):837--851, 2004.

\bibitem[CKPT17]{ChristensenKPT17}
Henrik~I. Christensen, Arindam Khan, Sebastian Pokutta, and Prasad Tetali.
\newblock Approximation and online algorithms for multidimensional bin packing:
  {A} survey.
\newblock {\em Comput. Sci. Rev.}, 24:63--79, 2017.

\bibitem[dlVL81]{VegaL81}
Wenceslas~Fernandez de~la Vega and George~S. Lueker.
\newblock Bin packing can be solved within 1+epsilon in linear time.
\newblock {\em Comb.}, 1(4):349--355, 1981.

\bibitem[DS14]{DinurS14}
Irit Dinur and David Steurer.
\newblock Analytical approach to parallel repetition.
\newblock In {\em Symposium on Theory of Computing, {STOC} 2014}, pages
  624--633, 2014.

\bibitem[Fei98]{Feige98}
Uriel Feige.
\newblock A threshold of ln \emph{n} for approximating set cover.
\newblock {\em J. {ACM}}, 45(4):634--652, 1998.

\bibitem[FK15]{FK15}
Alan Frieze and Michał Karoński.
\newblock {\em Introduction to Random Graphs}.
\newblock Cambridge University Press, 2015.

\bibitem[FKT89]{FaigleKT89}
Ulrich Faigle, Walter Kern, and Gy{\"{o}}rgy Tur{\'{a}}n.
\newblock On the performance of on-line algorithms for partition problems.
\newblock {\em Acta Cybern.}, 9(2):107--119, 1989.

\bibitem[GGJY76]{GareyGJ76}
M.~R. Garey, Ronald~L. Graham, David~S. Johnson, and Andrew~Chi{-}Chih Yao.
\newblock Resource constrained scheduling as generalized bin packing.
\newblock {\em J. Comb. Theory, Ser. {A}}, 21(3):257--298, 1976.

\bibitem[GJ79]{GareyJ79}
M.~R. Garey and David~S. Johnson.
\newblock {\em Computers and Intractability: {A} Guide to the Theory of
  NP-Completeness}.
\newblock W. H. Freeman, 1979.

\bibitem[GL18]{GL18}
Venkatesan Guruswami and Euiwoong Lee.
\newblock Strong inapproximability results on balanced rainbow-colorable
  hypergraphs.
\newblock {\em Comb.}, 38(3):547--599, 2018.

\bibitem[GS20]{GS20}
Venkatesan Guruswami and Sai Sandeep.
\newblock Rainbow coloring hardness via low sensitivity polymorphisms.
\newblock {\em {SIAM} J. Discret. Math.}, 34(1):520--537, 2020.

\bibitem[HS87]{HochbaumS87}
Dorit~S. Hochbaum and David~B. Shmoys.
\newblock Using dual approximation algorithms for scheduling problems
  theoretical and practical results.
\newblock {\em J. {ACM}}, 34(1):144--162, 1987.

\bibitem[HS19]{HarrisS19}
David~G. Harris and Aravind Srinivasan.
\newblock The {M}oser-{T}ardos framework with partial resampling.
\newblock {\em J. {ACM}}, 66(5):36:1--36:45, 2019.

\bibitem[IKKP19]{IKKP19}
Sungjin Im, Nathaniel Kell, Janardhan Kulkarni, and Debmalya Panigrahi.
\newblock Tight bounds for online vector scheduling.
\newblock {\em {SIAM} J. Comput.}, 48(1):93--121, 2019.

\bibitem[Jan10]{Jansen10}
Klaus Jansen.
\newblock An {EPTAS} for scheduling jobs on uniform processors: Using an {MILP}
  relaxation with a constant number of integral variables.
\newblock {\em {SIAM} J. Discret. Math.}, 24(2):457--485, 2010.

\bibitem[JS03]{JansenS03}
Klaus Jansen and Roberto Solis{-}Oba.
\newblock An asymptotic fully polynomial time approximation scheme for bin
  covering.
\newblock {\em Theor. Comput. Sci.}, 306(1-3):543--551, 2003.

\bibitem[Juh82]{Juhasz82}
Ferenc Juh{\'{a}}sz.
\newblock The asymptotic behaviour of {L}ov{\'{a}}sz' theta-function for random
  graphs.
\newblock {\em Comb.}, 2(2):153--155, 1982.

\bibitem[KAR00]{KAR00}
V.~S.~Anil Kumar, Sunil Arya, and H.~Ramesh.
\newblock Hardness of set cover with intersection 1.
\newblock In {\em Automata, Languages and Programming, 27th International
  Colloquium, {ICALP} 2000}, pages 624--635, 2000.

\bibitem[Kho01]{Khot01}
Subhash Khot.
\newblock Improved inaproximability results for maxclique, chromatic number and
  approximate graph coloring.
\newblock In {\em 42nd Annual Symposium on Foundations of Computer Science,
  {FOCS} 2001}, pages 600--609, 2001.

\bibitem[Knu94]{Knuth94}
Donald~E. Knuth.
\newblock The sandwich theorem.
\newblock {\em Electron. J. Comb.}, 1, 1994.

\bibitem[LY94]{LY94}
Carsten Lund and Mihalis Yannakakis.
\newblock On the hardness of approximating minimization problems.
\newblock {\em J. {ACM}}, 41(5):960--981, 1994.

\bibitem[Mos15]{Moshkovitz15}
Dana Moshkovitz.
\newblock The projection games conjecture and the {NP}-hardness of ln
  n-approximating set-cover.
\newblock {\em Theory Comput.}, 11:221--235, 2015.

\bibitem[MR10]{MR10}
Dana Moshkovitz and Ran Raz.
\newblock Two-query {PCP} with subconstant error.
\newblock {\em J. {ACM}}, 57(5):29:1--29:29, 2010.

\bibitem[MRT13]{MRT13}
Adam Meyerson, Alan Roytman, and Brian Tagiku.
\newblock Online multidimensional load balancing.
\newblock In {\em Approximation, Randomization, and Combinatorial Optimization.
  Algorithms and Techniques, {APPROX/RANDOM} 2013}, pages 287--302, 2013.

\bibitem[PTUW11]{panigrahy2011heuristics}
Rina Panigrahy, Kunal Talwar, Lincoln Uyeda, and Udi Wieder.
\newblock Heuristics for vector bin packing.
\newblock 2011.

\bibitem[Ray21]{ray2021aptas}
Arka Ray.
\newblock There is no {APTAS} for 2-dimensional vector bin packing: Revisited.
\newblock {\em CoRR}, abs/2104.13362, 2021.

\bibitem[Raz98]{Raz98}
Ran Raz.
\newblock A parallel repetition theorem.
\newblock {\em {SIAM} J. Comput.}, 27(3):763--803, 1998.

\bibitem[Spi94]{spieksma1994branch}
Frits~CR Spieksma.
\newblock A branch-and-bound algorithm for the two-dimensional vector packing
  problem.
\newblock {\em Computers \& operations research}, 21(1):19--25, 1994.

\bibitem[ST12]{ShachnaiT12}
Hadas Shachnai and Tami Tamir.
\newblock Approximation schemes for generalized two-dimensional vector packing
  with application to data placement.
\newblock {\em J. Discrete Algorithms}, 10:35--48, 2012.

\bibitem[Tre01]{Trevisan01}
Luca Trevisan.
\newblock Non-approximability results for optimization problems on bounded
  degree instances.
\newblock In {\em Proceedings on 33rd Annual {ACM} Symposium on Theory of
  Computing, 2001}, pages 453--461, 2001.

\bibitem[Woe97]{Woeginger97}
Gerhard~J. Woeginger.
\newblock There is no asymptotic {PTAS} for two-dimensional vector packing.
\newblock {\em Inf. Process. Lett.}, 64(6):293--297, 1997.

\end{thebibliography}
\bibliographystyle{alpha}

\appendix
\section{Hardness of simple $k$-set cover}
\label{sec:simple}

The hardness result of Kumar, Arya, and Ramesh~\cite{KAR00} is obtained from the Label Cover problem using a partition gadget along the lines of the reduction of Lund and Yannakakis~\cite{LY94}. The set families in the reduction in ~\cite{LY94} have large intersections. ~\cite{KAR00} get around this by using two main ideas: 
\begin{enumerate}
    \item They use a different partition system wherein each partition is a disjoint union of a large (super constant) number of sets instead of just $2$ sets in ~\cite{LY94}. 
    \item They use multiple sets for each label assignment to a vertex of the Label Cover, unlike a single set corresponding to each label of each vertex in~\cite{LY94}.
\end{enumerate}

As ~\cite{KAR00} were proving a $\Omega(\log n)$ hardness of the set cover, the universe size of the partition system is chosen to be the same as the number of vertices in the Label Cover instance.
This forces the set sizes to be very large. We can get around this issue by simply defining the partition system on a set of size $B$, where $B$ is a large constant. This also has an added benefit that we no longer require sub-constant hardness from the Label Cover instances, thus giving us NP-hardness directly. This observation is used by Trevisan~\cite{Trevisan01} to obtain $\ln B-O(\ln \ln B)$ NP-hardness of set cover on instances where each set has cardinality at most $B$, from Feige's $(1-\epsilon)\ln n$ set cover hardness~\cite{Feige98}. 

We now describe the parameter modifications in full detail. Let $B$ be a large constant. 

We start our reduction from Label Cover instances with soundness $\gamma = \frac{1}{32\beta^2\log^2B}$ where $\beta$ is an absolute constant to be fixed later.
\begin{theorem}(~\cite{ALMSS98,Raz98})
\label{thm:lc-constant-soundness}
Given a Label Cover instance defined on a bipartite graph $G=(V,E)$ with left alphabet $\Sigma_L$ and right alphabet $\Sigma_R$, it is NP-hard to distinguish between the following: 
\begin{enumerate}
    \item (Completeness). There exists a labeling $\sigma : V \rightarrow \Sigma_L \cup \Sigma_R$ that satisfies all the constraints. 
    \item (Soundness). No labeling to $V$ can satisfy more than $\gamma$ fraction of the constraints. 
\end{enumerate}
Furthermore the instances satisfy the following properties:
\begin{enumerate}
\item The alphabet sizes $d=|\Sigma_L|$ and $d'=|\Sigma_R|$ are both upper bounded by $(\log B)^{O(1)}$.
\item The maximum degree $deg$ of $G$ is upper bounded by $(\log B)^{O(1)}$.
\end{enumerate}
\end{theorem}

Following the convention in~\cite{KAR00}, we assume that the number of vertices on the left side in $G$ is equal to that on the right side of $G$, and we denote this number by $n'$.

We now construct a partition system $\mathcal{P}$ on a universe $N$ of size $B$.
The system $\mathcal{P}$ has $d' \times (deg+1) \times d$ partitions. Each partition has $m=B^{1-\epsilon}$ parts, where $\epsilon$ is a small constant to be fixed later. The partition system is divided into $d'$ groups each containing $(deg +1 )\times d$ partitions. Each group is further organized into $deg +1 $ subgroups each of which contains $d$ partitions. Let $P_{g,s,p}$ denote the $p$th partition in the $s$th subgroup of the $g$th group and $P_{g,s,p,k}$ denote the $k$th set in $P_{g,s,p}$ where $g\in [d'],s\in[deg +1],p\in[d],k\in[m]$.
The partition system satisfies the four properties in Section $4$ of~\cite{KAR00}, the only difference being that the universe $N$ now has size $B$ instead of $n'$. Thus, the covering property (Property $4$ in~\cite{KAR00}) now states that any covering of $N$ with $\beta m \log B$ sets should contain at least $\frac{3m}{4}$ sets from the same partition. 
Such a partition system is shown to exist for large enough $B$ in~\cite{KAR00} using a randomized construction. They also derandomize the construction. But for our setting, as $B$ is a constant, we just need to show the existence of such a partition system. 

We reduce the Label Cover instance in~\Cref{thm:lc-constant-soundness} to a set cover instance $\mathcal{SC}$ by the same construction as in~\cite{KAR00}: we have a partition system corresponding to each edge of the Label Cover instance, and the union of the elements in the partition systems is the element set of $\mathcal{SC}$. The sets in $\mathcal{SC}$, $C_k(v,a)$ are defined exactly as in~\cite{KAR00}. The cardinality of each set is at most $B' =deg \times B \leq B^2$. Each element is present in at most $md =O(B)$ sets. The fact that $\mathcal{SC}$ is a simple set system follows from Lemma 1 of~\cite{KAR00}. By Lemma 2 in~\cite{KAR00}, if there is a labeling of the Label Cover instance, then there is a set cover of size $n'm$ in $\mathcal{SC}$. If there is a set cover of size $\frac{\beta}{2}n'm\log B$ in $\mathcal{SC}$, then there is a labeling of $G$ that satisfies $\gamma$ fraction of constraints. The proof of this soundness follows along the same lines as Lemma $3$ of~\cite{KAR00}, with the only difference being that we now define the $\textit{good}$ edges as edges having $\#(e)\leq \beta m \log B$. 
\section{SDP Relaxation of Monochromatic-Clique}
\label{sec:appendix-vs}
We consider the following SDP relaxation of the graph coloring problem on $G=(V,E)$: 
\begin{align*}
    \text{Minimize } &k \\ 
    \langle u_i, u_i \rangle &= 1 \,\, \forall i \in V \\
    \langle u_i, u_j \rangle &\leq \frac{-1}{k-1} \, \, \forall (i,j)\in E
\end{align*}
The optimal solution to this SDP is referred to as the vector chromatic number $\chi_v(G)$ of the graph $G$. It is equivalent to the Lovasz theta function of the complement of $G$. We have the following sandwich property due to~\cite{Knuth94}:
\[
\omega(G) \leq \chi_v(G) \leq \chi(G)
\]

\subsection{Algorithm when $B> \sqrt{n}$}
There is a simple algorithm for the \mc$(k, B)$ problem when $B>k$: 
We compute $\chi_v(G)$ in polynomial time, and we check if $\chi_v(G)\leq k $. In this case, there is no clique of size $B$ in $G$, and we output YES. 
If $\chi_v(G)> k$, then the graph cannot be colored with $k$ colors, and in this case, we output NO. 

Note that if $k(B-1) \geq n$, there is always an assignment of $k$ colors to the vertices of the graph without a clique of size $B$, thus the problem is trivial.

\subsection{Integrality gap}

The above algorithm proves that in any graph with vector chromatic number at most $k$, there is an assignment of $k$ colors to the vertices that has monochromatic clique of size at most $\sqrt{n}$. We now prove that this cannot be significantly improved: 
\begin{theorem}
\label{thm:sdp-gap}
For $n$ large enough, there exists a graph $G=(V,E)$ with $n$ vertices, and a parameter $k$ such that
\begin{enumerate}
    \item $\chi_v(G)\leq k$. 
    \item In any assignment of $k$ colors to the vertices of $G$, there is a monochromatic clique of size $n^{\Omega(1)}$.
\end{enumerate}
\end{theorem}
\begin{proof}
We first prove the following: for large enough $n$, there exists a graph $G$ on $n$ vertices, and an integer $k$ such that 
\begin{enumerate}
    \item $\upvartheta(G) \leq k$. 
    \item In any assignment of $k$ colors to the vertices of the graph $G$, there exists a monochromatic independent set of size $B=n^{\Omega(1)}$.
\end{enumerate}
Our construction is a probabilistic one: we sample $G$ from $G(n,p)$ with $p=\frac{1}{\sqrt{n}}$. It has been proved~\cite{Juhasz82} that the Lovasz theta function of this random graph satisfies 
\[
\upvartheta(G) \leq 2 n^{\frac 34} + \Tilde{O}(n^{\frac 13}\log n)
\]
with high probability. 
We set $k = 3n^{\frac 34}$. For large enough $n$, with high probability, we have $\upvartheta(G)\leq k$. 

Furthermore, the random graph $G(n,p)$ with $p=o(n^{-\frac{2}{5}})$ has no $K_6$ with high probability(See e.g., \cite{FK15}). Thus, using~\Cref{lem:ramsey-independent}, we can infer that in any subset of size $\frac{n^{\frac{1}{4}}}{3}$, there is an independent set of size at least $\frac{n^{\frac{1}{24}}}{2}$. Hence, in any assignment of $k$ colors to the vertices of the graph $G$, there is a monochromatic independent set of size $n^{\Omega(1)}$. Taking the complement, we get a graph with the required properties. 
\end{proof}

\end{document}